\documentclass[12pt]{article}
\usepackage[usenames]{color}
\usepackage[ukrainian,english]{babel}
\usepackage[cp1251]{inputenc}
\usepackage{doc}
\usepackage{amsmath}
\usepackage{amsfonts}   
\usepackage{cite} 
\usepackage{fullpage}                           
\usepackage[lmargin=0.6in,rmargin=0.6in,tmargin=0.6in,headsep=.2in]{geometry}
\usepackage{graphicx}
\usepackage{forloop}
\usepackage{etoolbox}
\usepackage{calc}
\usepackage{wrapfig,lipsum,booktabs} 
\usepackage{xtab} 
\usepackage[dvipsnames]{xcolor}
\usepackage[normalem]{ulem}
\usepackage{sectsty}
\usepackage{amsmath,amsfonts,amssymb,amsthm,epsfig,epstopdf,url,array}
 \usepackage[retainorgcmds]{IEEEtrantools}
 \usepackage{makeidx,epsfig,lscape}
 \usepackage{xcolor,pict2e}
\usepackage{upgreek}
   
\makeatother
\usepackage{makeidx}
\makeindex
\makeatletter
\renewcommand{\@biblabel}[1]{#1.}

\newtheorem{note}{Remark}         
\newtheorem{te}{Theorem}           
\newtheorem{leme}{Lemma}         
\newtheorem{de}{Definition}          
\newtheorem{ce}{Consequence}    
\newtheorem{prope}{Proposition} 
\newtheorem{ax}{Axiom}  

\begin{document}
\centerline{ \bf\large{  Mode of sustainable economic development.  }}

\vskip 5mm

{\bf \centerline {\Large N.S. Gonchar \footnote{ This work is partially supported by the Fundamental Research Program of the Department of Physics and Astronomy of the National Academy of Sciences of Ukraine "Building and researching financial market models using the methods of nonlinear statistical physics and the physics of nonlinear phenomena N 0123U100362." }} }

\vskip 5mm
\centerline{\bf { Bogolyubov Institute for Theoretical Physics }} 
\centerline{\bf {of the National Academy of Sciences of Ukraine.}}
\vskip 2mm

\begin{abstract}
To implement the previously formulated principles of sustainable economic development, all non-negative solutions of the linear system of equations and inequalities, which are satisfied by the vector of real consumption, are completely described. It is established that the vector of real consumption with the minimum level of excess supply is determined by the solution of some quadratic programming problem.
The necessary and sufficient conditions are established  under which the economic system, described by  the "input-output" production model, functions in the mode  of sustainable development.
A complete description of the equilibrium states  for which markets are partially cleared in the economy model of production "input-output" is given, on the basis that all solutions of   system of linear equations  and inequalities are completely described.
The existence of a family of taxation vectors in the "input-output" model of production, under which the economic system is able to function in the mode of sustainable development, is proved. Restrictions were found for the vector of taxation in the economic system, under which the economic system is able to function in the mode of sustainable development.
The axioms of the aggregated description of the economy is proposed. 
\end{abstract}

\centerline{{\bf Keywords:} Technological mapping; Economic balance; }
\centerline{ Clearing markets; Vector of taxation;}
 \centerline{Sustainable economic development.}
 \centerline{Aggregated description of economy.}

\section{Introduction.}
The paper  details  the principles of economy sustainable development at the micro-economic level for the production model "input - output" and is a continuation of the papers \cite{11Gonchar11}  \cite{10Gonchar2}, \cite{11Gonchar2}.
This model of production is the basis for the formulation of concepts: gross output, direct costs, gross domestic product, gross added value, gross accumulation, export vector, import vector.
By sustainable economic development we understand the growth of the gross domestic product over a long period of time. During this period, the wages of employees increase, most companies operate in a profitable mode, and the exchange rate of the national currency is stable. In order to formulate this definition mathematically, such a concept as the consistency of the supply structure with the demand structure was introduced.
At the firm level, this would mean studying the structure of demand for manufactured goods in such a way that the volume of goods produced is consistent with the structure of demand. The latter will ensure, according to proven theorems, the appropriate profit for firms.

Why is the concept of sustainable economic development important? Due to the cyclical nature of economic development, it is necessary to understand the internal causes of this phenomenon. The concept of sustainable economic development serves this purpose. Under conditions of uncertainty, firms produce goods by taking loans from banks or issue securities to finance the necessary costs of production.
In order to repay loans, pay taxes on production, pay wages and ensure sufficient profit to restore fixed assets and expand production, it is necessary that the manufactured products be sold at prices at which the added value created is sufficient for this.

The model of sustainable development is formulated in such a way that for production technologies,  described by technological mappings from the KTM class, there exists an equilibrium price vector for which firms are able to purchase production materials and services for the production of the final product, and consumers of the final product to satisfy their needs.
In addition, it should be noted that there is a taxation system that ensures complete clearing of markets in a certain period of the economy's functioning. Therefore, the model of sustainable development is certain material and value balances, according to which the economic system is able to function continuously, creating a final product that is fully consumed in it and exported.

We call this mode of functioning of the economy the mode of sustainable development.

The purpose of the paper is to formulate the conditions under which the real economic system is able to function in the mode of sustainable development. Such conditions are limitations for the levels of taxation, under which the final product will be created in the economic system to satisfy the needs of consumers both in the sphere of the real economy and in the sphere of services. These are the restrictions on gross outputs, which are easily verified on the basis of Theorems established in section 3.

An economic system may not be in a mode of sustainable development due to the fact that the supply of produced goods is in excess. 
This is due to the fact that during the reporting period, not all the products produced can be sold according to the concluded contracts. These can be food industry products, agricultural products, stockpiles of manufactured weapons.
Therefore, the conditions of sustainable development, which provide for a complete clearing of the markets, will not be fully fulfilled. For this purpose, it is necessary to describe all states of economic equilibrium possible in the economic system, that is, to specify all equilibrium price vectors for which demand does not exceed supply.
To describe all equilibrium states, one should be able to solve systems of nonlinear equations and inequalities. Mathematical methods for this did not exist before this paper and the papers of \cite{11Gonchar11}, \cite{10Gonchar2} ,\cite{Gonchar2}.

In this paper, we  reduced this problem to the solution of two problems: the description of all solutions
 of a linear system of equations and inequalities and solving a linear homogeneous problem with respect to the equilibrium price vector for which partial clearing of the markets takes place.
 
 It should be noted that the description of all solutions
 linear system of equations and inequalities in this work was obtained for the first time.
 
 Each such equilibrium state is characterized by the level of excess supply. Its value ranges from zero to one. At zero level of excess supply, the markets are completely cleared. The greater the value of the level of excess supply, the closer the economic system is to the state of recession, see: \cite{4Gonchar}, \cite{3Gonchar}, \cite{5Gonchar} .
 
 Among all the described equilibrium states, there is one with the smallest excess supply.
 As a rule, the economy system
 is in that state.
To find it, you should solve the problem of quadratic programming, and based on this solution, construct the solution of the linear problem, finding the equilibrium state.
Then calculate the level of excess supply. Having done this for each reporting year, its dynamics should be investigated. If it grows, then this is an indicator that the economic system may slip into a state of recession.

Check whether the existing taxation system satisfies the proved inequalities that must be satisfied by an economy operating in a mode of sustainable development.

If it satisfies the derived inequalities, it should be checked whether the components of the final consumption vector are strictly positive. The presence of negative components can be associated with a negative trade balance, with the withdrawal of the created final product into the shadow sector of the economy.
In accordance with this, an economic policy can be developed to adjust its further development.

 The axiomatic of the aggregated description of the economy is proposed. On this basis, an aggregated matrix of direct costs, an aggregated vector of gross output, an aggregated gross added value, and an aggregated vector of final consumption were constructed. The main Theorems about the existence of the taxation system and limitations for the taxation system are also formulated for the aggregated description of the economy.
 
The proven Theorems show that in case of incomplete correspondence between the structure of demand and the structure of supply, there is a surplus of industrial goods, the price of which is not determined by demand and supply due to saturation of consumers with a certain group of goods. To characterize this phenomenon, an important concept of the vector of real consumption in the economic system is introduced. Based on this concept, the concept of the volume of unsold goods for a given period was introduced, which is defined as the relative value of unsold goods for a certain period.

 Section 2 gives the definition of a polyhedral cone and conditions for a given vector to belong this cone.
Theorem \ref{jant39} describes all strictly positive solutions of the system of equations, provided that the right-hand side of the equations belongs to the polyhedral cone formed by the columns of the matrix of the left-hand side of the system of equations.

The consistency of the structure of supply with the structure of demand gives Definition \ref{VickTin9}. The Definition \ref{1VickTin9} contains the consistency of the supply structure with the demand structure of a certain rank in the strict sense.
Lemma \ref{wickkteen1} contains sufficient conditions for consistency of property vectors and demand vectors in the strict sense. Lemma \ref{pupvittin7} asserts the existence of a strictly positive solution of a certain system of equations.
The Theorem \ref{wickkteen3}  formulates the necessary and sufficient conditions for the existence of a state of equilibrium under which the markets are completely cleared.
The Definition \ref{10VickTin9} contains the definition of consistency of the supply structure with the demand structure in a weak sense.

The definition of consistency of the supply structure with the demand structure in the weak sense of a certain rank is contained in Definition \ref{11VickTin9}. Theorem \ref{100wickkteen3} contains the necessary and sufficient conditions for the existence of a state of equilibrium, under which the markets are completely cleared under the condition that the supply structure is weakly consistent with the demand structure.

Theorems \ref{mykt1}, \ref{mykt6} and their consequences provide sufficient conditions for the existence of a solution to the system of inequalities. Theorem \ref{mykt19} establishes sufficient conditions for the existence of an equilibrium vector of prices. Sufficient conditions for the existence of economic equilibrium are contained in the Theorem \ref{mykt27}, which are conditions for the aggregate supply vector. The proven Theorems are a description of equilibrium states in the economic system.

Theorem \ref{TtsyVtsja1} gives necessary and sufficient conditions for the existence of solutions of the system of equations and inequalities, which are satisfied by the vector of real consumption.

The solutions of a certain system of linear equations and inequalities are completely described in the Theorem \ref{TVYA}.

The sufficient conditions for the existence of an equilibrium price vector are given in Theorem \ref{1TtsyVtsja4} provided that the vector of final consumption is a vector of real consumption generated by a certain non-negative  vector.

Theorem \ref{2TtsyVtsja4} is another variant of Theorem \ref{1TtsyVtsja4}.

The necessary and sufficient conditions for the existence of an equilibrium price vector corresponding to the vector of real consumption generated by some a certain non-negative vector are given in Theorem \ref{TtsyVtsja5}.

Theorem \ref{TtsyVtsja3} is a statement that for not every vector proposed for consumption for a given demand structure there exists an equilibrium price vector. That is, there is no market mechanism for the sale of manufactured goods in the long term.

In the definition \ref{kolja2}, it is stated that between every vector of real consumption, which is a solution of a system of equations and inequalities,  and a vector of equilibrium prices
there is a correspondence.

Theorem \ref{myktina19} is a justification of the above correspondences. In this case, there is an equilibrium price vector corresponding to market supply and demand. 

The necessary and sufficient conditions for the existence of an equilibrium price vector are given in Theorem \ref{MykHon1}.

On the basis of proven Theorems, the concept of the volume of unsold products in a given period is introduced for the corresponding equilibrium state in the case of partial clearing of markets.

In the the definition \ref{myktinavitka1}, it is stated that between the real consumption vector, generated by the quadratic programming problem, and   the equilibrium price vector there exists a correspondence. 

In section 3, Theorem \ref{1pupvittin13} gives the necessary and sufficient conditions for the  functioning of the economic system in the mode of sustainable development described by the "input- output" production model. 

In Theorem \ref{tsytsjatintsytsja1} for the "input - output" model, the existence of a taxation system under which the price vector formed in the economic system is an equilibrium price vector is established. 

The best taxation system has been built, under which the economic system is able to function in the mode of sustainable development.

In Theorem \ref{main1}, restrictions are found for the vector of taxation, under which a final product can be created in the economic system that ensures the functioning of the economic system in  the mode of sustainable development.

In Theorem \ref{2gonpupvittin3}, a certain system of taxation is built, under which the economic system is able to function in the mode of sustainable development. In this case, the vector of gross output satisfies a certain system of equations, the solution of which always exists and which determines the vector of final consumption.

Section 4 is an illustration of the application of the results obtained in Sections 2 and 3 to the study of real economic systems described in value indicators.
Theorem \ref{2main1} is a reformulation of Theorem \ref{main1} for the case of real economic systems.
It asserts the restrictions for the taxation systems in real economic systems under which the economic system is able to operate in a sustainable development mode.
Theorem \ref{4main1} is a partial case of Theorem \ref{2main1}.

Section 5 contains the results for the aggregate description of the economy.

\section{Construction of equilibrium states with excess supply.}

Here and further, $R_+^n \setminus \{0\}$ is a cone
formed from the non-negative orthant 
$R_+^n $ by discarding the null vector
 $\{0\}=\{0, \ldots, 0\}.$ Next, the cone 
 $R_+^n \setminus \{0 \}$ is denoted by 
 $ R_+^n .$

We  give a number of definitions useful for the future.
\begin{de}\label{mant0}
Under the polyhedral non
negative cone formed by the set of vectors 
$\{a_i,\ i=\overline {1,t}\}$ of the $n$-dimensional space $R^n$ we understand the set of vectors of the form
\begin{eqnarray*} 
d = \sum\limits_{i=1}^t\alpha_i a_i,\end{eqnarray*}
where $\alpha=\{\alpha_i\}_{i=1}^t$ runs through the set $R_+^t.$

\begin{de}\label{mant1}
The dimension of the non negative polyhedral cone formed by a set of vectors
$\{a_i, \ i=\overline {1,t}\}$ in the $n$-dimensional space $R^n$ is the maximum number of linearly independent vectors from the set of vectors 
$\{a_i,\ i=\overline {1,t}\}.$
\end{de}

\begin{de}\label{1mant2}
The vector $b $ belongs to the interior of the non negative polyhedral $r$-dimensional cone $r \leq n,$ created by the set of vectors
 $\{a_1,\dots ,a_t\}$ in the $n$-dimensional vector space $R^n$
 if there exists a strictly positive vector
  $\alpha=\{\alpha_i\}_{i=1}^t \in R_+^t$ such that
\begin{eqnarray*} 
b =\sum\limits _{s=1}^t a_s\alpha_s,\end{eqnarray*}
where $\alpha_s>0, \ s=\overline{1,t}.$
\end{de}
\end{de}

We present the necessary and sufficient conditions under which a certain vector belongs to the interior of a polyhedral cone.

\begin{leme}\label{mant2}
Let $\{a_1,\dots ,a_m\},$ $1\leq m\leq n,$ be a set of linearly independent vectors in $R_+^n.$ The necessary and sufficient conditions for the vector $b $ to belong to the interior of a non-negative cone $K_a^+$ formed by vectors
$\{a_i,\ i=\overline {1,m}\}$ are conditions
\begin{eqnarray} \label{mant3}
\langle f_i, b \rangle >0, \quad i=\overline {1,m}, \quad \langle f_i, b \rangle=
0,\quad i=\overline {m+1,n},
\end{eqnarray}
where $f_i,\ i=\overline {1,n},$ is a set of vectors biorthogonal to a set of linearly independent vectors $\bar a_i, \ i=\overline {1,n},$ and
$\bar a_i=a_i, \ i=\overline {1,m}.$
\end{leme}

For the proof of Lemma \ref{mant2}, see: \cite{Gonchar2}, \cite{7Gonchar}.
We will now describe the algorithm for constructing strictly positive solutions of the system of equations.
\begin{eqnarray} \label{luda05}
 \psi=\sum\limits_{i=1}^lC_i y_i, \quad y_i>0, \quad i=\overline{1,l},
\end{eqnarray}
relative to the vector $y=\{y_i\}_{i=1}^l$
or the same system of equations in coordinate form
\begin{eqnarray} \label{gonl40}
\sum\limits _{i=1}^l
c_{ki}y_i =\psi _k,\quad k=\overline{1,n},
\end{eqnarray}
for a vector
$\psi =\{\psi _1,\dots ,\psi _n\}$ belonging to the interior of the polyhedral cone formed by the vectors $\{C_i=\{c_{ki}\}_{k=1}^ n, \ i=\overline {1,l}\}.$
\begin{te}\label{jant39}
If a certain vector
$\psi$ belonging to the interior of the non-negative $r$-dimensional polyhedral cone formed by the vectors $\{C_i=\{c_{ki}\}_{k=1}^n,\ i=\overline {1 ,l}\},$
such that there exists a subset $r$ of linearly independent vectors of the set of vectors
$\{C_i, \ i=\overline {1,l}\},$ such that the vector $\psi$ belongs to the interior of the cone formed by this subset of vectors,
then there exist $l-r+1$ linearly independent nonnegative solutions $z_i$ of the system of equations (\ref{gonl40}) such that the set of strictly positive solutions of the system of equations (\ref{gonl40} )
is given by the formula
\begin{eqnarray} \label{gonl41} y=\sum\limits
_{i=r}^l\gamma _iz_i , \end{eqnarray}
where \begin{eqnarray*} z_i=\{\langle\psi
,f_1\rangle - \langle C_i,f_1\rangle y_i^*,\dots , \langle\psi
,f_r\rangle - \langle C_i,f_r\rangle y_i^*,0,\dots ,y_i^*, 0,\dots
,0\},\quad i=\overline {r+1,l},\end{eqnarray*}
\begin{eqnarray*} z_r=\{\langle\psi ,f_1\rangle ,\dots
,\langle\psi ,f_r\rangle ,0,\dots ,0\},\end{eqnarray*}
 \begin{eqnarray*} y_i^*=\left\{\begin{array}{ll}
  \min\limits_{s\in K_i}\frac {\langle \psi
,f_s \rangle} {\langle C_i,f_s \rangle },& K_i=\{s,\langle
C_i,f_s \rangle >0\},\\ 1, & {\langle C_i,f_s \rangle}\leq 0,
~\forall \ s=\overline{1,r}, \end{array} \right.
\end{eqnarray*}
and the components of the vector
$\{\gamma_i\}_{i=r}^l$ satisfy a set of inequalities
\begin{eqnarray*} \sum\limits _{i=r}^l\gamma _i=1, \quad
\gamma _i>0, \quad i=\overline {r+1,l}, \end{eqnarray*}
\begin{eqnarray} \label{eq2}
\sum\limits_{i=r+1}^l{\langle C_i,f_k \rangle} y_i^*\gamma_i<
{\langle\psi,f_k\rangle}, \quad k=\overline{1,r}.
\end{eqnarray}
\end{te}
For the proof of Theorem \ref{jant39}, see:\cite{Gonchar2}, \cite{7Gonchar}.
\begin{de}\label{VickTin9}
Let $C_i=\{c_{s i} \}_{s=1}^n \in R_+^n, \ i=\overline{1,l},$ be a set of demand vectors and $b_i= \{b_ {s i} \}_{s=1}^n \in R_+^n, \ i=\overline{1,l},$ be a set of supply vectors. We say that the structure of supply is consistent with the structure of demand in the strict sense
if the matrix $B$ has the representation $B=C B_1$, where the matrix $B$ consists of vectors $b_i \in R_+^n, \ i=\overline{1,l},$ as columns, and the matrix $C$ consists of vectors $C_i \in R_+^n, \ i=\overline{1,l},$ as columns, and $B_1$ is a square non-negative non-decomposable matrix.
\end{de}

\begin{de}\label{1VickTin9}
Let $C_i=\{c_{s i} \}_{s=1}^n \in R_+^n, \ i=\overline{1,l},$ be a set of demand vectors and $b_i= \{b_ {s i} \}_{s=1}^n \in R_+^n, \ i=\overline{1,l},$ be a set of supply vectors. We  say that the supply structure is consistent with the demand structure in the strict sense of the rank $|I|$ if there exists such a subset $I \subseteq N$ that the matrix $B^I$ has the representation $ B^I=C^I B_1 ^I$, where the matrix $B^I$ consists of vectors $b_i^I \in R_+^{|I|}, \ i=\overline{1,l },$ in the form of columns, and the matrix $C^ I$ consists of vectors $C_i^I\in R_+^n, \ i=\overline{1,l},$ in the form of columns, and $B_1^I= |b_{is}^{1,I}| _{i,s=1}^l$ is a square non-negative non-decomposable matrix, where $b_i^I=\{b_{ki}\}_{k \in I },$
$C_i^I=\{c_{ki}\}_{k \in I}$ and, in addition, the inequalities hold
\begin{eqnarray*}\label{2VickTin9}
 \sum\limits_{i=1}^l c_{ki}y_i^I < \sum\limits_{i=1}^l b_{ki}, \quad k \in N \setminus I, \quad y_i^I =\sum\limits_{s=1}^l b_{is}^{1,I}.
\end{eqnarray*}
\end{de}

\begin{leme}\label{wickkteen1}
Suppose that the set of supply vectors $b_i = \{b_{s i} \}_{s=1}^n\in R_+^n, \ i=\overline{1,l}, $ belongs to the polyhedral cone formed by the set of vectors of demand $\{C_i=\{c_{ki}\}_{k=1}^n \in R_+^n\ i=\overline {1,l}\}.$
Then for the matrix $B=|b_{ki}|_{k=1,i=1}^{n,l}$ formed by the columns of vectors $b_i=\{b_{ki}\}_{k =1} ^n \in R_+^n, \ i=\overline{1,l}, $ has a representation
\begin{eqnarray} \label{wickkteen2}
B=C B_1,
\end{eqnarray}
where the matrix $C=|c_{ki}|_{k=1,i=1}^{n,l}$ is formed by the columns of vectors $C_i=\{c_{ki}\}_{k =1}^n \in R_+^n, \ i=\overline{1,l}, $ and the matrix $B_1=
|b_{m i}^1|_{m=1,i=1}^{l}$ is non-negative. If, in addition, the set of supply vectors $b_i \in R_+^n, \ i=\overline{1,l}, $ belongs to the interior of the polyhedral cone formed by the demand vectors $\{C_i=\{c_{ki}\}_{k=1}^n \in R_+^n\ i=\overline {1,l} \},$ matrix $B_1$ can be chosen strictly positive.
\end{leme}

\begin{proof}
The proof of the first part of  Lemma
\ref{wickkteen1} follows from the second one.
If each vector $b_i, \ i=\overline{1,l}, $ belongs to the interior of the polyhedral cone formed by the vectors $C_i, \ i=\overline{1,l}, $  then according to the Theorem \ref{jant39} there exists a strictly positive vector $y_i=\{y_{ki} \}_{k=1}^l$ such that
$$ b_{ki}=\sum\limits_{s=1}^l c_{ks} y_{si}, \quad k=\overline{1,n}.$$
Denote $y_{si}=b_{si}^1,$ then we get
$$ b_{ki}=\sum\limits_{s=1}^l c_{ks} b_{si}^1, \quad k=\overline{1,n}, \quad i=\overline{1, l}.$$ This  proves Lemma \ref{wickkteen1}.
\end{proof}

\begin{leme}\label{pupvittin7}
Let $B_1=||b_{ki}^1||_{k,i=1}^l$ be a square non-negative non-decomposable matrix. Then the problem
\begin{eqnarray}\label{pupvittin8}
\sum\limits_{k=1}^l b_{ki}^1d_k=\sum\limits_{s=1}^l b_{is}^1 d_i, \quad i=\overline{1,l},
\end{eqnarray}
has a strictly positive solution with respect to the vector $d=\{d_k\}_{k=1}^l.$
\end{leme}

\begin{proof}
Due to the fact that the matrix $B_1$ is indecomposable, we have $\sum\limits_{s=1}^l b_{is}^1>0, \ i=\overline{1,l}.$
 Let's consider the problem
\begin{eqnarray}\label{pupvittin9}
\sum\limits_{k=1}^l e_{ki} d_k^1= d_i^1,\quad i=\overline{1,l},
\end{eqnarray}
and prove that it has a strictly positive solution with respect to the vector $d^1=\{d_k^1\}_{k=1}^l,$
where we introduced the denotation
$$e_{ki}=\frac{b_{ki}^1}{\sum\limits_{s=1}^l b_{ks}^1}.$$
To prove this, consider the problem
\begin{eqnarray}\label{pupvittin10}
\frac{d_i^1+\sum\limits_{k=1}^l e_{ki} d_k^1}{2}= d_i^1,\quad i=\overline{1,l},
\end{eqnarray}
in the set
$P=\{d^1=\{d_k^1\}_{k=1}^l , \ d_k^1 \geq 0, \ k=\overline{1,l}, \sum\limits_{k =1}^l d_k^1=1 \}.$
Due to equalities $\sum\limits_{i=1}^l e_{ki}=1, \ k=\overline{1,l},$
the map $ H(d^1)=\{H_i(d^1)\}_{i=1}^l, $
 maps $P$ into itself and is continuous on it, where
 $$ H_i(d^1)=\frac{d_i^1+\sum\limits_{k=1}^l e_{ki} d_k^1}{2},\quad i=\overline{1,l} .$$ According to Brauer's Theorem \cite{Nirenberg}, there exists a fixed point of the mapping $H(d^1),$ i.e
$$ \frac{d_i^1+\sum\limits_{k=1}^l e_{ki} d_k^1}{2}=d_i^1,\quad i=\overline{1,l}.$$

The same fixed point satisfies the problem (\ref{pupvittin9}). Since the matrix
 $E=||e_{ki}||_{k, i=1}^l$ is non-negative and indecomposable, we have $E^{l-1}>0.$ From the fact that $Ed ^1= d^1$ implies $E^{l-1}d^1=d^1.$ This proves that the vector $d^1$ is strictly positive. Lemma \ref{pupvittin7} is proved.
\end{proof}
\begin{te}\label{wickkteen3}
Let the supply structure agree with the demand structure in the strict sense with supply vectors $b_i=\{b_{ki}\}_{k=1}^n \in R_+^n, \ i=\overline{ 1,l}, $ and demand vectors $\{C_i=\{c_{ki}\}_{k=1}^n \in R_+^n\ i=\overline {1,l}\}, $ and let
 $\sum\limits_{s=1}^nc_{si}>0, i=\overline{1,l}, $ \ $\sum\limits_{i=1}^l c_{si}>0, s =\overline{1,n}.$ The necessary and sufficient conditions for the existence of a solution to the system of equations
\begin{eqnarray}\label{awickkteen4}
\sum\limits_{i=1}^l c_{ki}\frac{<b_i, p>}{<C_i, p>}=\sum\limits_{i=1}^l b_{ki}, \quad k=\overline{1,n},
\end{eqnarray}
 relative to the vector $p$ are that the vector $D=\{d_i\}_{i=1}^l$ belongs to the polyhedral cone formed by the vectors $C_k^T=\{c_{ki}\}_{i=1}^ l, \ k=\overline{1,n},$ where
 $B=C B_1,$ $B_1=| b_{ki}^1|_{k,i=1}^l$ is a non-negative non-decomposable matrix,
vector $D=\{d_i\}_{i=1}^l$ is a strictly positive solution of the system of equations
\begin{eqnarray}\label{awickkteen5}
\sum\limits_{k=1}^l b_{ki}^1d_k=y_i d_i, \quad y_i=\sum\limits_{k=1}^l b_{ik}^1, \quad i=\overline{ 1,l}.
\end{eqnarray}
\end{te}
The proof of Theorem \ref{wickkteen3}, see \cite{7Gonchar}.
\begin{de}\label{10VickTin9}
Let $C_i \in R_+^n, \ i=\overline{1,l},$ be a set of demand vectors and $b_i \in R_+^n, \ i=\overline{1,l}, $ be a set supply vectors. We say that the structure of supply is consistent with the structure of demand in a weak sense,
if the representation $B=C B_1$ is true for the matrix $B$, where the matrix $B$ consists of vectors $b_i \in R_+^n, \ i=\overline{1,l},$ as columns, and the matrix $C$ consists of vectors $C_i \in R_+^n, \ i=\overline{1,l},$ as columns, and $B_1$ is a square matrix satisfying the conditions
\begin{eqnarray}\label{10wickkteen33}
\sum\limits_{s=1}^lb_{is}^1 \geq 0,\quad i=\overline{1,l} \quad B_1=|b_{is}^1|_{i, s=1 }^l.
\end{eqnarray}
\end{de}

\begin{de}\label{11VickTin9}
Let $C_i \in R_+^n, \ i=\overline{1,l},$ be a set of demand vectors, and $b_i \in R_+^n, \ i=\overline{1,l}, $ set of supply vectors. We  say that the supply structure is consistent with the demand structure in the weak sense of the rank $|I|$, if there is such a subset $I \subseteq N$ that the representation $ B^I=C^I B_1$ is true for the matrix $B^I$ , where the matrix $B^I$ consists of vectors $b_i^I \in R_+^{|I|}, \ i=\overline{1,l },$ in the form of columns, and the matrix $C^ I$ consists of vectors $C_i^I\in R_+^n, \ i=\overline{1,l},$ in the form of columns and $B_1^I$ is a square matrix satisfying the conditions
\begin{eqnarray}\label{101wickkteen33}
\sum\limits_{s=1}^lb_{is}^{1, I} \geq 0,\quad i =\overline{1,l}, \quad B_1^I=|b_{is}^{1 ,I}|_{i,s=1}^l,
\end{eqnarray}
 where $b_i^I=\{b_{ki}\}_{k \in I},$
$C_i^I=\{c_{ki}\}_{k \in I}$ and, in addition, the inequalities hold
\begin{eqnarray*}\label{2VickTin9}
 \sum\limits_{i=1}^l c_{ki}y_i^I < \sum\limits_{i=1}^l b_{ki}, \quad k \in N \setminus I, \quad y_i^I =\sum\limits_{s=1}^l b_{is}^{1,I}.
\end{eqnarray*}
\end{de}

\begin{te}\label{100wickkteen3}
Let the supply structure agree with the demand structure in the weak sense with supply vectors $b_i=\{b_{ki}\}_{k=1}^n \in R_+^n, \ i=\overline {1,l}, $ and demand vectors $\{C_i=\{c_{ki}\}_{k=1}^n \in R_+^n\ i=\overline {1,l}\} ,$ and let
 $\sum\limits_{s=1}^nc_{si}>0, i=\overline{1,l}, $ \ $\sum\limits_{i=1}^l c_{si}>0, s =\overline{1,n}.$ The necessary and sufficient conditions for the existence of a solution to the system of equations
\begin{eqnarray}\label{wickkteen4}
\sum\limits_{i=1}^l c_{ki}\frac{<b_i, p>}{<C_i, p>}=\sum\limits_{i=1}^l b_{ki}, \quad k=\overline{1,n},
\end{eqnarray}
are that the vector $D=\{d_i\}_{i=1}^l$ belongs to the polyhedral cone created by the vectors $C_k^T=\{c_{ki}\}_{i=1}^l , \ k= \overline{1,n},$ where
 $B=C B_1,$ $B_1=| b_{ki}^1|_{k,i=1}^l$ is a square matrix,
vector $D=\{d_i\}_{i=1}^l$ is a strictly positive solution of the system of equations
\begin{eqnarray}\label{wickkteen5}
\sum\limits_{k=1}^l b_{ki}^1d_k=y_i d_i, \quad y_i=\sum\limits_{k=1}^l b_{ik}^1\geq 0, \quad i= \overline{1,l}.
\end{eqnarray}
\end{te}

\begin{proof} The proof of Theorem \ref{100wickkteen3} is the same as Theorem \ref{wickkteen3}.

Below we present the algorithms for constructing non-negative  solutions of a linear system of inequalities.
\end{proof}
\begin{te}\label{mykt1} Let $A =||a_{ki}||_{k,i=1}^n$ be a nonnegative non-decomposable matrix, and let the  vector $b=\{b_k\} _{k=1}^n$ be strictly positive.
In the set $Y=\{y=\{y_1,\ldots,y_n\}, \ y_i \geq 0, i=\overline{1,n}, \sum\limits_{i=1}^n y_i=1 \}$ there exists a solution to the system of equations
\begin{eqnarray}\label{mykt2}
\frac{y_k+ y_k\sum\limits_{i=1}^n\frac{a_{ki}}{b_k}y_i+\varepsilon}{1+ \sum\limits_{k=1}^n\sum\limits_{ i=1}^ny_k\frac{a_{ki}}{b_k}y_i+n \varepsilon }=y_k, \quad k=\overline{1,n},
\end{eqnarray}
which is strictly positive for every $\varepsilon>0.$
\end{te}

\begin{proof}
Consider the mapping $\varphi^{\varepsilon}(y)=\{\varphi_k^{\varepsilon}(y)\}_{k=1}^n$, where
\begin{eqnarray}\label{mykt3}
\varphi_k^{\varepsilon}(y)=\frac{y_k+y_k \sum\limits_{i=1}^n\frac{a_{ki}}{b_k}y_i+\varepsilon}{1+ \sum\limits_ {k=1}^n\sum\limits_{i=1}^n y_k\frac{a_{ki}}{b_k}y_i+n \varepsilon } \quad k=\overline{1,n}.
\end{eqnarray}
For any $\varepsilon>0$, it is continuous on the set $Y=\{y=\{y_1,\ldots,y_n\}, \ y_i \geq 0, i=\overline{1,n}, \sum\limits_{i=1}^n y_i=1\}$ and maps it into itself. Based on Brauer's Theorem \cite{Nirenberg}, there exists a fixed point  of the map $\varphi^{\varepsilon}(y)$ in the set $Y,$ that is,
\begin{eqnarray}\label{mykt4}
\varphi^{\varepsilon}(y^{\varepsilon})=y^{\varepsilon}.
\end{eqnarray}
Let us prove the strict positivity of $y^{\varepsilon}=\{y_k^{\varepsilon}\}_{k=1}^n. $ The  estimate
\begin{eqnarray}\label{mykt5}
y_k^{\varepsilon} \geq \frac{\varepsilon}{1+ \sum\limits_{k=1}^n\sum\limits_{i=1}^ny_k^{\varepsilon}\frac{a_{ki }}{b_k}y_i^{\varepsilon}+n \varepsilon }\geq \frac{\varepsilon}{1+ \sum\limits_{k=1}^n\sum\limits_{i=1}^n\frac{a_{ki}}{b_k}+n \varepsilon }>0
\end{eqnarray}
is valid.
Theorem \ref{mykt1} is proved.
\end{proof}

\begin{te}\label{mykt6} Let $A =||a_{ki}||_{k,i=1}^n$ be a nonnegative non-decomposable matrix, and the vector $b=\{b_k\} _{k=1}^n$ is strictly positive.
In the set $Y=\{y=\{y_1,\ldots,y_n\}, \ y_i \geq 0, i=\overline{1,n}, \sum\limits_{i=1}^n y_i=1 \}$ there exists a solution $y_0=\{y_k^0\}_{k=1}^n$ of the system of equations
\begin{eqnarray}\label{mykt7}
\frac{y_k+ y_k\sum\limits_{i=1}^n\frac{a_{ki}}{b_k}y_i}{1+ \sum\limits_{k=1}^n\sum\limits_{i= 1}^ny_k\frac{a_{ki}}{b_k}y_i }=y_k, \quad k=\overline{1,n}.
\end{eqnarray}
For those $k,$ for which $y_k^0>0,$ the equalities

\begin{eqnarray}\label{mykt8}
\frac{1+ \sum\limits_{i=1}^n\frac{a_{ki}}{b_k}y_i^0}{1+ \sum\limits_{k=1}^n\sum\limits_{ i=1}^ny_k^0\frac{a_{ki}}{b_k}y_i^0 }=1, \quad k=\overline{1,n},
\end{eqnarray}
are true and for those $k,$ for which $y_k^0=0,$ the inequalities

\begin{eqnarray}\label{mykt9}
\frac{1+ \sum\limits_{i=1}^n\frac{a_{ki}}{b_k}y_i^0}{1+ \sum\limits_{k=1}^n\sum\limits_{ i=1}^ny_k^0\frac{a_{ki}}{b_k}y_i^0 }\leq 1, \quad k=\overline{1,n},
\end{eqnarray}
are valid.
\end{te}
\begin{proof}
Based on Theorem \ref{mykt1}, for every $ \varepsilon >0$ there exists a solution $y^{\varepsilon}=\{y_k^{\varepsilon}\}_{k=1}^n. $ of the system of equations (\ref{mykt2}), which is strictly positive. Due to the compactness of the set $Y$, there exists a subsequence $ \varepsilon_m >0,$ such that $y^{\varepsilon_m}=\{y_k^{\varepsilon_m}\}_{k=1}^n $
goes to $y^{0}=\{y_k^{0}\}_{k=1}^n \in Y.$ Let the index $k $ be such that $y_k^0>0,$ then

\begin{eqnarray}\label{mykt10}
\frac{y_k^0+ y_k^0\sum\limits_{i=1}^n\frac{a_{ki}}{b_k}y_i^0}{1+ \sum\limits_{k=1}^n \sum\limits_{i=1}^ny_k^0\frac{a_{ki}}{b_k}y_i^0 }=y_k^0.
\end{eqnarray}
After reducing by $y_k^0>0,$ we get the equality (\ref{mykt8}).
Now let the index $k $ be the one for which $y_k=0.$
Then the inequality follows from the system of equations (\ref{mykt2}) for every $\varepsilon>0.$
\begin{eqnarray}\label{mykt11}
\frac{y_k^\varepsilon+ y_k^\varepsilon\sum\limits_{i=1}^n\frac{a_{ki}}{b_k}y_i^\varepsilon}{1+ \sum\limits_{k=1} ^n\sum\limits_{i=1}^ny_k^\varepsilon\frac{a_{ki}}{b_k}y_i^\varepsilon+n \varepsilon }\leq y_k^\varepsilon, \quad k=\overline{ 1,n}.
\end{eqnarray}
After reducing by $ y_k^\varepsilon>0,$ we get the inequality

\begin{eqnarray}\label{mykt12}
\frac{1+ \sum\limits_{i=1}^n\frac{a_{ki}}{b_k}y_i^\varepsilon}{1+ \sum\limits_{k=1}^n\sum\limits_ {i=1}^ny_k^\varepsilon\frac{a_{ki}}{b_k}y_i^\varepsilon+n \varepsilon }\leq 1.
\end{eqnarray}
Going to the limit in the inequality (\ref{mykt12}), we get the inequality

\begin{eqnarray}\label{mykt13}
\frac{1+ \sum\limits_{i=1}^n\frac{a_{ki}}{b_k}y_i^0}{1+ \sum\limits_{k=1}^n\sum\limits_{ i=1}^ny_k^0\frac{a_{ki}}{b_k}y_i^0 }\leq 1.
\end{eqnarray}
Theorem \ref{mykt6} is proved.
\end{proof}

\begin{ce}\label{mykt14} Let the quadratic form
$\sum\limits_{k=1}^n\sum\limits_{i=1}^ny_k\frac{a_{ki}}{b_k}y_i$ be strictly positive on the set $Y$. Then there exists a solution $z_0=\{z_i^0\}_{i=1}^n,$ of the system of inequalities
\begin{eqnarray}\label{mykt15}
\sum\limits_{i=1}^na_{ki} z_i^0 \leq b_k, \quad k=\overline{1,n},
\end{eqnarray}
 where $z_i^0=\frac{y_i^0}{\sum\limits_{k=1}^n\sum\limits_{i=1}^ny_k^0\frac{a_{ki}}{b_k}y_i ^0 },$ and
 $y_0=\{y_i^0\}_{i=1}^n$ is a solution of the system of equations (\ref{mykt7}).
\end{ce}

\begin{ce}\label{mykt16}
Let $I$ be a  set of
  indices $k \in N_0 =[1,2,\ldots,n]$ such that $z_k^0>0,$ and $J=N_0\setminus I.$ The vector $z_0^I=\{z_k^0 \}_{k\in I}$ satisfies the system of equations
\begin{eqnarray}\label{mykt17}
\sum\limits_{i \in I}a_{ki} z_i^0 =b_k, \quad k\in I,
\end{eqnarray}
and inequalities
\begin{eqnarray}\label{mykt18}
\sum\limits_{i \in I}a_{ki} z_i ^0 < b_k, \quad k\in J.
\end{eqnarray}
\end{ce}
It is evident that the set $I $ is non empty one.
\begin{te}\label{mykt19}
Let $I$ be a nonempty subset of the set $N_0,$  and the minor $||a_{ij}||_{i \in I, j \in I}$ of the nonnegative matrix $||a_{ij}||_{ i,j =1}^n$ is indecomposable and the vector $b$ is strictly positive. The system of equations
\begin{eqnarray}\label{mykt20}
\sum\limits_{i \in I}a_{ki} \frac{p_i b_i}{\sum\limits_{s \in I}a_{si}p_s} =b_k, \quad k\in I,
\end{eqnarray}
and inequalities
\begin{eqnarray}\label{mykt21}
\sum\limits_{i \in I}a_{ki} \frac{p_i b_i}{\sum\limits_{s \in I}a_{si}p_s} < b_k, \quad k\in J,
\end{eqnarray}
is always solvable in the set of strictly positive solutions with respect to the vector $p^I=\{p_i\}_{i\in I}$. There is one to one  correspondence between the solutions of the system of equations and inequalities (\ref{mykt20}), (\ref{mykt21}) and the solutions of the system of equations and inequalities (\ref{mykt17}), (\ref{mykt18}) from Corollary \ref{mykt16}.
\end{te}
\begin{proof}
Let $I$ be a nonempty set. If $p_0^I=\{p_i^0\}_{i\in I}$ is a strictly positive solution of the system of equations and inequalities (\ref{mykt20}), (\ref{mykt21}), then introducing the denotation
 $$ z_i^0=\frac{p_i^0 b_i}{\sum\limits_{s \in I}a_{si}p_s^0}, \ i \in I,$$ 
 we get a proof in one direction, i.e.
$z_0^I=\{z_i^0\}_{i\in I}$ is a strictly positive solution of the system of equations and inequalities (\ref{mykt17}), (\ref{mykt18}) from Corollary \ref{mykt16}. Let us first prove that there is no state of economic equilibrium under the condition that $ I $ is an empty set.

From the opposite. Let there exist a nonzero solution $p=\{p_i\}_{i=1}^n$ of the system of inequalities
\begin{eqnarray}\label{2mykt21}
\sum\limits_{i=1}^n a_{ki} \frac{p_i b_i}{\sum\limits_{s=1}^n a_{si}p_s} < b_k, \quad k=\overline{1 ,n}.
\end{eqnarray}
Introduce the denotation $y_i=\frac{p_i}{\sum\limits_{s=1}^n a_{si}p_s}, \ i=\overline{1,n}.$
Then the nonzero solution $p=\{p_i\}_{i=1}^n$ satisfies the system of equations

\begin{eqnarray}\label{3mykt21}
p_i= y_i \sum\limits_{s=1}^n a_{si}p_s, \quad i =\overline{1,n}.
\end{eqnarray}
 In the $n$-dimensional space $R^n$ of vectors $p=\{p_i\}_{i=1}^n$ we introduce the norm
$ ||p||= \sum\limits_{s=1}^n b_s|p_s| $ and estimate the norm of the operator
$[AY]^T,$ where $Y=||\delta_{ij}y_i||_{i,j=1}^n.$ We have
$$||AYp||=\sum\limits_{i=1}^n b_i y_i |\sum\limits_{s=1}^n a_{si}p_s|
\leq \max\limits_{s} \sum\limits_{i=1}^n b_i y_i\frac{a_{si}}{b_s} ||p||.$$
Hence $||[AY]^T||\leq \max\limits_{s} \sum\limits_{i=1}^n b_i y_i\frac{a_{si}}{b_s}<1.$
The last inequality holds due to the fulfillment of inequalities (\ref{2mykt21}).
This means that the system of equations (\ref{3mykt21}) has only zero solution. Contradiction.

 Therefore, let $I$ be a nonempty set and there exists a strictly positive solution $z_0^I=\{z_i^0\}_{i\in I}$ of the system of equations and inequalities (\ref{mykt17}), (\ref {mykt18}) from Corollary \ref{mykt16}. Let's put
\begin{eqnarray}\label{kolja1}
 z_i^0=\frac{p_i b_i}{\sum\limits_{s \in I}a_{si}p_s},\quad \ i \in I.
\end{eqnarray}
With respect to the vector $p^I=\{p_i\}_{i\in I}$, we obtain a system of equations
\begin{eqnarray}\label{mykt22}
p_i=y_i \sum\limits_{s \in I}a_{si}p_s, \quad \ i \in I,
\end{eqnarray}
where the denotations $y_i=\frac{z_i^0}{b_i}, \ i \in I,$  are introduced. Let us  prove the existence of a strictly positive solution of the system of equations (\ref{mykt22}).

Consider the nonlinear mapping $\varphi(p)=\{\varphi_i(p)\}_{i \in I}$

\begin{eqnarray}\label{mykt23}
\varphi_i(p)=\frac{p_i+y_i \sum\limits_{s \in I}a_{si}p_s}{1+\sum\limits_{i \in I}y_i \sum\limits_{s \in I}a_{si}p_s}, \quad \ i \in I,
\end{eqnarray}
on the set $P=\{p=\{p_i\}_{i\in I}, \ p_i\geq 0, \ \sum\limits_{i \in I}p_i=1\}.$ The mapping
 $\varphi(p)$ is a continuous one of the set $P$ into itself. According to Brauer's Theorem \cite{Nirenberg}, there exists a fixed point $p_0^I=\{p_i^0\}_{i\in I}$ of this mapping in the set $P$, that is,
  \begin{eqnarray}\label{mykt24}
\frac{p_i^0+y_i \sum\limits_{s \in I}a_{si}p_s^0}{1+\sum\limits_{i \in I}y_i \sum\limits_{s \in I} a_{si}p_s^0}=p_i^0, \quad \ i \in I.
\end{eqnarray}
This fixed point satisfies the system of equations
\begin{eqnarray}\label{mykt25}
\lambda p_i^0=y_i \sum\limits_{s \in I}a_{si}p_s^0, \quad \ i \in I,
\end{eqnarray}
where $ \lambda=\sum\limits_{i \in I}y_i \sum\limits_{s \in I}a_{si}p_s^0.$
Let us show that $ \lambda=1.$ Multiplying both the left and right parts of the equalities by $b_i$ and summing over $i \in I,$ we get
\begin{eqnarray}\label{mykt26}
\lambda \sum\limits_{i \in I} p_i^0 b_i=\sum\limits_{s \in I}(\sum\limits_{i \in I}a_{si}z_i^0)p_s^0= \sum\limits_{s \in I}p_s^0 b_s.
\end{eqnarray}
Due to the fact that $b_i>0, \ i \in I,$ \ $p_0^I\neq 0$ we  have $\sum\limits_{s \in I}p_s^0 b_s>0.$ From the equality (\ref{mykt26}) we get that $ \lambda=1.$

The strict positivity of the vector $p_0^I=\{p_i^0\}_{i\in I}$ follows from the non-negativity of the matrix
$||a_{ij}||_{i,j =1}^n$ and indecomposability of the minor $||a_{ij}||_{i \in I, j \in I}.$ Indeed, denote by $ A^I(y)$  the operator whose components are given by the formulas $[A^I(y)p]_i=y_i \sum\limits_{s \in I}a_{si}p_s,\ i \in I.$ Then the system of equations (\ref{mykt25}) in operator form can  be written as follows
$p_0^I=A^I(y)p_0^I.$ Hence the vector $p_0^I$ satisfies the system of equations $p_0^I=[A^I(y)]^{|I|-1} p_0^I ,$ where $| I|$ is the power of the set $I.$ Due to the indecomposability of the minor $||a_{ij}||_{i \in I, j \in I},$
the matrix corresponding to the operator $[A^I(y)]^{|I|-1}$ is strictly positive. From here we get the strict  positivity of $p_0^I.$ Theorem \ref{mykt19} is proved.
\end{proof}

\begin{te}\label{mykt27} Let $A=||a_{ij}||_{i,j =1}^n$ be a nonnegative matrix and $I$ be a nonempty set.
For any strictly positive vector $\bar b \leq b$  $\bar b_i=b_i, i \in I,$ $\bar b_i<b_i, i \in J, $ $ I\neq \emptyset,$ $J= N_0 \setminus I, $ and such that it belongs to the interior of the cone formed by the column vectors $a_i=\{a_{ki}\}_{k=1}^n, \ i\in I,$ of the non-negative matrix $||a_{ij}||_{i,j =1}^n,$ a
  minor $||a_{ij}||_{i \in I, j \in I}$ of which  is indecomposable, there exists a state of economic equilibrium, that is, there is an equilibrium price vector $p_0=\{p_i^0\}_{i= 1}^n$ such that
\begin{eqnarray}\label{mykt28}
\sum\limits_{i =1}^na_{ki} \frac{p_i^0 b_i}{\sum\limits_{s=1}^na_{si}p_s^0} \leq b_k, \quad k=\overline{1,n}.
\end{eqnarray}
\end{te}
\begin{proof}
Under the conditions of Theorem \ref{mykt27}
\begin{eqnarray}\label{mykt29}
\bar b=\sum\limits_{i \in I}a_{ki}z_i^0, \quad k=\overline{1,n},
\end{eqnarray}
 where $z_i^0>0, \ i \in I.$
Therefore, there exists a solution of the system of equations and inequalities (\ref{mykt17}), (\ref{mykt18}) from Corollary \ref{mykt16}, i.e.
\begin{eqnarray}\label{mykt30}
\sum\limits_{i \in I}a_{ki}z_i^0= b_k, \quad k \in I.
\end{eqnarray}
\begin{eqnarray}\label{mykt31}
\sum\limits_{i \in I}a_{ki}z_i^0 < b_k, \quad k \in J.
\end{eqnarray}
Based on Theorem \ref{mykt19}, there exists an equilibrium vector of prices. Theorem \ref{mykt27} is proved.
\end{proof}

Theorem \ref{mykt27} provides a partial description of all equilibrium states of the economic system described by the "input-output" model. Note that the vector $b-\bar b$, if it is nonzero, does not find consumers in the market of goods and services. Vector
$\bar b$, as before, will be called the vector of real consumption.
\begin{te}\label{TtsyVtsja1}
Let $b$ be a strictly positive vector that does not belong to the positive cone formed by the column vectors of the non negative matrix $A.$
The necessary and sufficient conditions of the solution existence of
  the system of equations and inequalities
\begin{eqnarray}\label{100kolja}
 \sum\limits_{j=1}^n a_{ij}z_j=b_i, \quad i \in I,
\end{eqnarray}
\begin{eqnarray}\label{100kolja1}
\sum\limits_{j=1}^n a_{ij}z_j<b_i, \quad i \in J,
\end{eqnarray}
where $I$ is a nonempty set,  is a condition: there exists a non-negative vector $y =\{y_i\}_{i=1}^n, $ $Ay\neq 0$ and such that
$$a \sum\limits_{j=1}^n a_{ij}y_j=b_i, \quad i \in I,$$
 where $a=\min\limits_{1\leq i\leq n}\frac{b_i}{\sum\limits_{j=1}^n a_{ij}y_j}.$
\end{te}

\begin{proof}
Necessity. Let there exist a solution of the system of equations (\ref{100kolja}) and inequalities (\ref{100kolja1}) of Theorem \ref{TtsyVtsja1}. Then
$$a=\min\limits_{1\leq i\leq n}\frac{b_i}{\sum\limits_{j=1}^n a_{ij}z_j}=1,$$
 and, in addition,
$$\frac{b_i}{\sum\limits_{j=1}^n a_{ij}z_j}=1,\quad i \in I,$$

$$\frac{b_i}{\sum\limits_{j=1}^n a_{ij}z_j}>1,\quad i \in J.$$
It is obvious that $Az \neq 0.$
The necessity is established.

Sufficiency. For any non-negative vector $y =\{y_i\}_{i=1}^n,$ $Ay\neq 0$ let
$a=\min\limits_{1\leq i\leq n}\frac{b_i}{\sum\limits_{j=1}^n a_{ij}y_j}.$ Then $a< \infty.$ We put
 $$I=\left\{i, \frac{b_i}{\sum\limits_{j=1}^n a_{ij}y_j}=a\right\}.$$
Then $I$ is a nonempty set and the inequalities hold
$$ \sum\limits_{j=1}^n a_{ij}z_j=b_i, \quad i \in I,$$
$$ \sum\limits_{j=1}^n a_{ij}z_j<b_i, \quad i \in J,$$
where $z=\{z_i\}_{j=1}^n,$ $z_i=a y_i, i=\overline{1,n}.$
Theorem \ref{TtsyVtsja1} is proved.
\end{proof}

We give a complete description of the non-negative  solutions of the system of equations and inequalities (\ref{mykt17}), (\ref{mykt18}). We denote by $a_i=\{a_{ki}\}_{i=1}^n, \ i=\overline{1,n},$ the $i$-th column of the matrix $A.$
Let us consider the numbers $d_i=\min\limits_{1\leq k \leq n}\frac{b_k}{a_{ki}}, \ i=\overline{1,n}.$
\begin{te}\label{TVYA} Let the strictly positive vector $b$ not belong to the cone formed by the column vectors $a_i=\{a_{ki}\}_{i=1}^n, \ i=\overline{ 1,n},$ of the non-negative non-decomposable matrix $A.$ Any non-negative solution of the system of equations (\ref{mykt17}) and inequalities (\ref{mykt18}) is given by the formula
$$ z=\{a(\alpha)\alpha_i d_i\}_{i=1}^n,$$
where $\alpha=\{\alpha_i \}_{i=1}^n \in Q=
\{\alpha=\{\alpha_i \}_{i=1}^n, \ \alpha_i \geq 0, \sum\limits_{i=1}^n\alpha_i=1\},$
$$a(\alpha)=\min\limits_{1\leq k \leq n}
\frac{b_k}{[\sum\limits_{i=1}^n\alpha_i d_i a_i]_k}\geq 1.$$
The function $ a(\alpha)$ is bounded and continuous on the set $Q.$
\end{te}

\begin{proof}
Let $z_0=\{z_i^0\}_{i=1}^n$ be a certain vector that is a solution of the system of equations (\ref{mykt17}) and inequalities (\ref{mykt18}).
Let's denote
$$\alpha_i=\frac{\frac{z_i^0}{d_i}}{\sum\limits_{j=1}^n\frac{z_j^0}{d_j} }\quad i=\overline{1 ,n}.$$
Then
$$ Az_0=\sum\limits_{i=1}^n a_iz_i^0=\sum\limits_{j=1}^n\frac{z_j^0}{d_j} \sum\limits_{i=1}^ n \alpha_i d_i a_i.$$
Because of
$$ \min\limits_{1\leq k \leq n} \frac{b_k}{[Az_0]_k }=1,$$
we get
$$\sum\limits_{j=1}^n\frac{z_j^0}{d_j}= \min\limits_{1\leq k \leq n} \frac{b_k}{[\sum\limits_{i =1}^n \alpha_i d_i a_i]_k}.$$
It is obvious that, conversely, every vector $\alpha=\{\alpha_i \}_{i=1}^n \in Q$ corresponds to
solution of the system of equations (\ref{mykt17}) and inequalities (\ref{mykt18}), which is given by the formula
$$ z=\{a(\alpha)\alpha_i d_i\}_{i=1}^n,$$
where
$$a(\alpha)=\min\limits_{1\leq k \leq n}
\frac{b_k}{[\sum\limits_{i=1}^n\alpha_i d_i a_i]_k}.$$

Let us establish that $a(\alpha) \geq 1.$
It is obvious that $a_i d_i\leq b.$ Multiplying by $\delta_i\geq 0$ the left and right parts of the last inequality and summing over $i$ and assuming that $\sum\limits_{i=1}^n\delta_i >0$ we will get
$$\frac{\sum\limits_{i=1}^n\delta_i d_i a_i}{\sum\limits_{i=1}^n\delta_i }\leq b.$$
Denoting $\alpha_i=\frac{\delta_i}{\sum\limits_{i=1}^n\delta_i}, \ i=\overline{1,n},$ we get what we need.
It follows from the indecomposability of the matrix $A$ that for every index $1\leq i \leq n$ there exists an index $k$ such that
$$ \sum\limits_{j=1}^na_{kj}z_j^0 \geq a_{ki}z_i^0, $$
where $a_{ki}>0.$ Hence
$$ z_i^0 \leq \frac{b_k}{a_{ki}}\leq \frac{\max\limits_{1\leq k\leq n}b_k}{\min\limits_{a_{ki}>0 }a_{ki}}=c_0<\infty.$$
Due to the arbitrariness of the solution $z_0=\{z_i^0\}_{i=1}^n$ of the system of equations (\ref{mykt17}) and inequalities (\ref{mykt18}), we obtain
$$ a(\alpha)\alpha_i d_i \leq c_0.$$
Or
$$ a(\alpha)\alpha_i \leq \frac{ c_0}{d_i}.$$
After summing over the index $i$, we get
$$ a(\alpha) \leq \sum\limits_{i=1}^n\frac{ c_0}{d_i}.$$
The boundedness of $ a(\alpha)$ is established.

Let's prove the continuity of $ a(\alpha)$.
Let's prove the continuity of $ a(\alpha)$.
Let us consider the functions $\sum\limits_{i=1}^na_{ki}\alpha_i d_i, \ k=\overline{1,n}$ Every of these  function is continuous on the set $P.$ Since  $a(\alpha)$ is bounded let us denote $B=\sup\limits_{\alpha \in P}a(\alpha)<\infty.$
For sufficiently small $\varepsilon >0$
that satisfies inequalities $\frac{b_k}{\varepsilon}>B, \ k=\overline{1,n},$ let us introduce the sets $$A_k^{\varepsilon}=\left\{\alpha \in P, \ \sum\limits_{i=1}^na_{ki}\alpha_i d_i\leq \varepsilon \right\}\quad k=\overline{1,n}.$$
If the  set $A_k^{\varepsilon}$ is nonempty
we introduce the function
 \begin{eqnarray}\label{won10}
 V_k^{\varepsilon}(\alpha)=
 \left\{\begin{array}{l l} \frac{b_k}{\sum\limits_{i= 1}^na_{ki}\alpha_i d_i}, & \mbox{if} \quad \alpha \in P \setminus A_k^{\varepsilon},\\
\frac{b_k}{\varepsilon}, & \mbox{if} \quad \alpha \in  A_k^{\varepsilon}.
\end{array} \right.
\end{eqnarray}
If the set $A_k^{\varepsilon}$ is empty one
we  put 
$$ V_k^{\varepsilon}(\alpha)= \frac{b_k}{\sum\limits_{i= 1}^na_{ki}\alpha_i d_i}.$$
The functions $ V_k^{\varepsilon}(\alpha)$  are continuous on the set $P$ and the equality 
$$ \min\limits_{1\leq k\leq n}V_k^{\varepsilon}(\alpha)=a(\alpha)$$
is true. Really, from the inequalities
$$ \frac{b_k}{\sum\limits_{i= 1}^na_{ki}\alpha_i d_i}\geq   V_k^{\varepsilon}(\alpha), \quad k=\overline{1,n} $$
it follows that
$$ \min\limits_{1\leq k\leq n}V_k^{\varepsilon}(\alpha)\leq a(\alpha).$$ The inverse inequality follows from the note that if for a certain point 
$\alpha \in P$ we obtain that 
$$ \min\limits_{1\leq k\leq n}V_k^{\varepsilon}(\alpha)< a(\alpha).$$ 
From the definition of $ V_k^{\varepsilon}(\alpha)$ it means that  $\min\limits_{1\leq k\leq n}V_k^{\varepsilon}(\alpha)=\frac{ b_{k_0}}{\varepsilon}>B$ for a certain $k_0.$ But this is impossible. The function $ \min\limits_{1\leq k\leq n}V_k^{\varepsilon}(\alpha)$ is continuous on $P.$

Theorem \ref{TVYA} is proved.
\end{proof}
\begin{prope}\label{g1}
 In the set of solutions $Z_0 = \{z_0=\{z_i^0\}_{i=1}^n=\{a(\alpha)\alpha_i d_i\}_{i=1}^n, \ \alpha \in Q\}$ of the system of equations (\ref{mykt17}) and inequalities (\ref{mykt18}) there exists a minimum of the function
$$W(\alpha)=\sum\limits_{k=1}^n [b_k- a(\alpha)\sum\limits_{i=1}^na_{ki}\alpha_i d_i]^2.$$
This minimum is global on the set of all solutions of the system of inequalities (\ref{mykt15}),
i.e
$$\min\limits_{\alpha \in Q}W(\alpha)=\min\limits_{z \in Z} \sum\limits_{k=1}^n [b_k- \sum\limits_{i= 1}^na_{ki}z_i]^2,$$
where $Z$ is the set of all non-negative solutions of the system of inequalities (\ref{mykt15}).
\end{prope}
\begin{proof}
The function $W(\alpha)$ is continuous on the closed bounded set $Q,$ because so is the function $a(\alpha)$ due to its continuity. According to the Weierstrass Theorem, there exists a minimum of the function
$W(\alpha).$
For any solution $z=\{z_i\}_{i=1}^ n \in Z$
let's denote
$$\alpha_i=\frac{\frac{z_i}{d_i}}{\sum\limits_{j=1}^n\frac{z_j}{d_j} }\quad i=\overline{1,n}. $$
Then
$$ Az=\sum\limits_{i=1}^n a_iz_i=\sum\limits_{j=1}^n\frac{z_j}{d_j} \sum\limits_{i=1}^n \alpha_i d_i a_i \leq b.$$
From here
$$ \sum\limits_{j=1}^n\frac{z_j}{d_j} \leq \min\limits_{1\leq k \leq n}\frac{b_k}{[\sum\limits_{i= 1}^n \alpha_i d_i a_i]_k}=a(\alpha).$$
Therefore
$$ \sum\limits_{k=1}^n [b_k- \sum\limits_{i=1}^na_{ki}z_i]^2 \geq \sum\limits_{k=1}^n [b_k- a(\alpha)\sum\limits_{i=1}^na_{ki}\alpha_i d_i]^2\geq $$
$$\min\limits_{\alpha \in Q}\sum\limits_{k=1}^n [b_k- a(\alpha)\sum\limits_{i=1}^na_{ki}\alpha_i d_i] ^2.$$
Taking the minimum of $z \in Z$, we have
$$\min\limits_{z \in Z} \sum\limits_{k=1}^n [b_k- \sum\limits_{i=1}^na_{ki}z_i]^2\geq \min\limits_ {\alpha \in Q}\sum\limits_{k=1}^n [b_k- a(\alpha)\sum\limits_{i=1}^na_{ki}\alpha_i d_i]^2.$$
The inverse inequality is obvious due to the inclusion $Z \supset Z_0, $ where $Z_0$ is the set of solutions of the system of equations (\ref{mykt17}) and inequalities (\ref{mykt18}).
Proposition \ref{g1} is proved.
\end{proof}
Note that the set of non-negative solutions $Z$ of the system of equations and inequalities (\ref{mykt15}) is a closed bounded convex set. Indeed, the boundedness of the set of solutions follows from the indecomposability of the matrix $A$ and the finiteness of the vector $b.$ The convexity of the set of solutions is obvious. Then such a set is a convex combination of its extreme points. Such extreme points are some subset of the solutions constructed in Theorem \ref{TVYA}, which do not have common equalities.
Below we establish a number of theorems that will allow us to put into correspondence the  vector of prices to every non-negative vector
 $z=\{z_i\}_{i=1}^n,$ that satisfies the system of equations and inequalities
\begin{eqnarray}\label{100kolja2}
\sum\limits_{j=1}^n a_{ij} z_j=b_i, \quad i \in I,
\end{eqnarray}
\begin{eqnarray}\label{kolja5}
 \sum\limits_{j=1}^n a_{ij} z_j<b_i, \quad i \in J,
\end{eqnarray}
where $I$ is a nonempty set.

\begin{te}\label{1TtsyVtsja4}
Let $A=||a_{ij}||_{ij=1}^n$ be a strictly positive matrix, and $z=\{z_i\}_{i=1}^n$ be a non-negative nonzero vector. Then there exists a solution to the system of equations
\begin{eqnarray}\label{101kolja1}
 \sum\limits_{i=1}^n a_{ki}\frac{\bar b_i p_i}{\sum\limits_{s=1}^n a_{si}p_s}=\bar b_k \quad k=\overline{1,n},
\end{eqnarray}
in the set $P=\{p=\{p_i\}_{i=1}^n,\ p_i \geq 0,\ i=\overline{1,n}, \ \sum\limits_{i=1} ^n p_i=1\},$ where $\bar b=Az=\{\bar b_i\}_{i=1}^n.$
\end{te}
\begin{proof}
Consider the system of nonlinear equations
$$\frac{p_i+\frac{z_i}{\bar b_i}\sum\limits_{s =1}^n a_{si}p_s}{1+\sum\limits_{i =1}^n\frac{ z_i}{\bar b_i}\sum\limits_{s=1 }^n a_{si}p_s}=p_i, \quad i=\overline{1,n},$$
in the set $P=\{p=\{p_i\}_{i=1}^n, \ p_i\geq 0, \ i=\overline{1,n}, \ \sum\limits_{i=1 } ^n p_i=1\}.$
The left part of this system of equations is a continuous mapping of the set $P$ into itself. According to Brouwer's Theorem, there exists a fixed point
$p_0=\{p_i^0\}_{i=1}^n$ of the mapping given by the left part of this system of equations. Or
$$z_i\sum\limits_{s =1}^n a_{si}p_s^0=\lambda \bar b_i p_i^0, \quad i=\overline{1,n},$$
where
$\lambda=\sum\limits_{i =1}^n\frac{z_i}{\bar b_i}\sum\limits_{s=1}^na_{si}p_s^0.$
Summing over the index $i$, we will have the left and right parts of the above equalities

$$\sum\limits_{i =1}^n \bar b_i p_i^0=\lambda \sum\limits_{i =1}^n \bar b_i p_i^0.$$
  It is obvious that the vector $p_0$ is strictly positive, and $\sum\limits_{i =1} ^n \bar b_i p_i^0>0,$ therefore $\lambda=1$ and
$$ z_i=\frac{\bar b_i p_i^0}{ \sum\limits_{s=1}^n a_{si}p_s^0}, \quad i=\overline{1,n}. $$
Using the fact that
$$ \sum\limits_{i=1}^n a_{ij}z_j=\bar b_i, \quad i=\overline{1,n},$$
and substituting $z_i$ in the equality above, we get the required result.
This proves Theorem \ref{1TtsyVtsja4}.
\end{proof}

It can be proved similarly
\begin{te}\label{2TtsyVtsja4}
Let $A=||a_{ij}||_{ij=1}^n$ be a positive indecomposable matrix, and $z=\{z_i\}_{i=1}^n$ be a strictly positive vector. Then there is a solution to the system of equations
$$ \sum\limits_{i=1}^n a_{ki}\frac{\bar b_i p_i}{\sum\limits_{s=1}^n a_{si}p_s}=\bar b_k \quad k =\overline{1,n}, $$
in the set $P=\{p=\{p_i\}_{i=1}^n,\ p_i \geq 0,\ i=\overline{1,n}, \ \sum\limits_{i=1} ^n p_i=1\},$ where $\bar b=Az=\{\bar b_i\}_{i=1}^n.$
\end{te}

\begin{te}\label{TtsyVtsja5}
Let $A=||a_{ij}||_{ij=1}^n$ be a nonnegative nonzero matrix. The necessary and sufficient conditions for the existence of a solution of the system of equations
$$ \sum\limits_{i=1}^n a_{ki}\frac{\bar b_i p_i}{\sum\limits_{s=1}^n a_{si}p_s}=\bar b_k \quad k =\overline{1,n}, $$
in the set $P$ is the existence of a solution to the system of equations
$$ z_i=\frac{\bar b_i p_i^0}{ \sum\limits_{s=1}^n a_{si}p_s^0}, \quad i=\overline{1,n}, $$
in the set $P$ for some non negative nonzero vector $z=\{z_i\}_{i=1}^n$ and $\bar b=Az=\{\bar b_i\}_{i=1} ^n.$
\end{te}
\begin{proof}
The proof is obvious.
\end{proof}
\begin{te}\label{TtsyVtsja3}
Let $A=||a_{ij}||_{ij=1}^n$ be a non-negative nonzero matrix and let $z=\{z_i\}_{i=1}^n$ be a non-negative vector, which is a solution to the system of inequalities
\begin{eqnarray}\label{100kolja3}
 \sum\limits_{i=1}^n a_{ij}z_j=b_i, \quad i \in I,
\end{eqnarray}
\begin{eqnarray}\label{100kolja4}
\sum\limits_{i=1}^n a_{ij}z_j<b_i, \quad i \in J,
\end{eqnarray}
where $I$ and $J$ are nonempty sets. If $b$ is a strictly positive vector, and $z_j>0$ for some $j \in J,$ then there is no solution to the system of equations
\begin{eqnarray}\label{100kolja5}
 z_i=\frac{b_i p_i}{ \sum\limits_{s=1}^n a_{si}p_s}, \quad i=\overline{1,n},
\end{eqnarray}
in the set $P=\{p=\{p_i\}_{i=1}^n,\ p_i \geq 0,\ i=\overline{1,n}, \ \sum\limits_{i=1} ^n p_i=1\}.$
\end{te}

\begin{proof}
Argument from the opposite. Suppose that there is a nonzero nonnegative vector $ p_0=\{p_i^0\}_{i=1}^n$ which is a solution of the system of equations of Theorem \ref{TtsyVtsja3}. Then $\sum\limits_{s=1}^n a_{si}p_s^0\neq 0$ and
$$z_i \sum\limits_{s=1}^n a_{si}p_s^0=b_i p_i^0,\quad i=\overline{1,n}.$$ Hence we have $ p_j ^0>0$ for $j \in J$ specified in Theorem \ref{TtsyVtsja3}. Summing up the left and right parts by $i$, we get
$$ \sum\limits_{s=1}^n ( \sum\limits_{i=1}^n a_{si}z_j) p_s^0= \sum\limits_{i=1}^n b_i p_i^0 .$$
But
$$ \sum\limits_{s=1}^n ( \sum\limits_{i=1}^n a_{si}z_j) p_s^0< \sum\limits_{i=1}^n b_i p_i^0 .$$
Contradiction. This proves the theorem \ref{TtsyVtsja3}.
\end{proof}

\begin{de}\label{kolja2}
Let $A$ be a nonnegative matrix, and the vector $b$ is strictly positive, which does not belong to the cone formed by the column vectors of the matrix $A. $ We  say that the vector 
$z=\{z_i\}_{i=1}^n,$ which is a solution of the system of inequalities
\begin{eqnarray}\label{100kolja6}
 \sum\limits_{j=1}^n a_{ij}z_j=b_i, \quad i \in I,
\end{eqnarray}
\begin{eqnarray}\label{100kolja7}
 \sum\limits_{j=1}^n a_{ij}z_j<b_i, \quad i \in J,
\end{eqnarray}
where $I$ is a nonempty set, corresponds to the equilibrium price  vector   if there exists a solution to the system of equations
\begin{eqnarray}\label{100kolja8}
 \sum\limits_{i=1}^n a_{ki}\frac{\bar b_i p_i}{\sum\limits_{s=1}^n a_{si}p_s}=\bar b_k \quad k=\overline{1,n},
\end{eqnarray}
in the set $P=\{p=\{p_i\}_{i=1}^n,\ p_i \geq 0,\ i=\overline{1,n}, \ \sum\limits_{i=1} ^n p_i=1\},$ where $\bar b=Az=\{\bar b_i\}_{i=1}^n.$
\end{de}
\begin{ce}\label{5TtsyVtsja4}
There is no equilibrium state corresponding to the vector $z=\{z_i\}_{i=1}^n$ which is a solution of the system of equations (\ref{100kolja5}) of Theorem \ref{TtsyVtsja3} with respect to the vector $ p=\{p_i\}_{ i=1}^n.$
\end{ce}
\begin{proof}
Indeed, if it were not so, then the solution of the system of equations (\ref{100kolja5}) of Theorem \ref{TtsyVtsja3} with respect to the vector $ p=\{p_i\}_{i=1}^n$ should exist. But it is not so.
\end{proof}
 Theorem \ref{myktina19} is the basis for the determining of the equilibrium price vector in the case of partial clearing of markets.

\begin{te}\label{myktina19}
Let $||a_{ij}||_{i,j =1}^n$ be a nonnegative matrix, $b=\{b_i\}_{=1}^n$ be a strictly positive vector,  and let the matrix  $||a_{ij}||_{i,j \in I}$ be an indecomposable one for a certain non empty set $I$. 
Then the equilibrium   price vector $p=\{p_i\}_{i=1}^n$,  being a solution  of the system of equations and inequalities
 $$  \sum\limits_{j=1}^n  a_{ij}\frac{b_j p_j}{\sum\limits_{s=1}^n a_{sj}p_s}=b_i \quad i \in I,     $$
\begin{eqnarray}\label{mari1}
   \sum\limits_{j=1}^n a_{ij}\frac{b_j p_j}{\sum\limits_{s=1}^n a_{sj}p_s}<b_i \quad i \in J,     
\end{eqnarray}
in the set $P=\{p=\{p_i\}_{i=1}^n,\ p_i \geq 0,\ i=\overline{1,n}, \ \sum\limits_{i=1} ^n p_i=1\}$
is a solution of the system of equations
\begin{eqnarray}\label{mari2}
z_i=\frac{\bar b_i p_i}{\sum\limits_{s=1}^n a_{si}p_s}, \quad i=\overline{1,n},
\end{eqnarray}
where $\bar b=\{\bar b_i\}_{i=1}^n,$ $\bar b=A z,$ the vector $z=\{z_i\}_{i=1}^n$ is determined as follows
 $z_i=z_i^0, \ i \in I, z_i=0, \ i \in J.$ The non-negative  vector $z_0^I=\{z_i^0\}_{i \in I}$ satisfies the system of equations and inequalities
\begin{eqnarray}\label{100kolja9}
   \sum\limits_{j \in I} a_{ij} z_j^0=b_i \quad i \in I,
\end{eqnarray}
\begin{eqnarray}\label{100kolja10}
  \sum\limits_{j \in I} a_{ij}z_j^0<b_i \quad i \in J.
\end{eqnarray}
 \end{te}
\begin{proof}
Let there exist a solution of the system of equations and inequalities (\ref{mari1}) with respect to the vector $p=\{p_i\}_{i=1}^n$ in the set $P.$ Let us show that $p_i=0, \ i \in J.$ 
We  lead the proof from the opposite.
Let at least one component $p_k, k \in J,$ of the vector $p$ be strictly positive.
Then, multiplying by $p_i, i=\overline{1,n},$ the $i$-th equation or inequality and summing the left and right parts, respectively, we get the inequality
$ \sum\limits_{j=1}^n b_j p_j< \sum\limits_{i=1}^n b_i p_i. $
Due to the strict positivity  of the vector $b$, this inequality is impossible, because $ \sum\limits_{j=1}^n b_j p_j>0.$ Therefore, our assumption is not correct, and therefore $p_i=0, \ i \in J. $
The remaining component $p_i, \  i \in I,$ of the vector $p$ is a solution of the system of equations and inequalities
\begin{eqnarray}\label{vasja}
  \sum\limits_{j\in I} a_{ij}\frac{b_j p_j}{\sum\limits_{s\in I } a_{sj}p_s}=b_i \quad i \in I,
 \end{eqnarray}
 \begin{eqnarray}\label{vasja1}
  \sum\limits_{j\in I} a_{ij}\frac{b_j p_j}{\sum\limits_{s\in I } a_{sj}p_s}<b_i \quad i \in J.
 \end{eqnarray}
Let us  introduce the denotation
 \begin{eqnarray}\label{vasja1} 
z_j^0= \frac{b_j p_j}{\sum\limits_{s\in I } a_{sj}p_s}, \quad j \in I.
 \end{eqnarray}
It is evident that  the  equalities and inequalities 
 $$ \sum\limits_{j \in I} a_{ij} z_j^0=b_i \quad i \in I, $$
 $$ \sum\limits_{j \in I} a_{ij}z_j^0<b_i \quad i \in J, $$
 are valid.
If we introduce a vector
$z=\{z_i\}_{i=1}^n$ where $z_i=z_i^0, \ i \in I, z_i=0, \ i \in J,$
then we obtain a system of equations and inequalities
$$ \sum\limits_{j =1}^na_{ij} z_j=b_i \quad i \in I, $$
 $$ \sum\limits_{j=1} ^na_{ij}z_j<b_i \quad i \in J. $$
The price vector
$ p=\{p_i\}_{i=1}^n$ is a solution of the system of equations (\ref{mari2}), due to the fact that $\bar b=Az,$ and therefore is also a solution of the system equations
\begin{eqnarray}\label{mari3}
   \sum\limits_{j =1}^n a_{ij}\frac{\bar b_j p_j}{\sum\limits_{s=1}^n a_{sj}p_s}=\bar b_i \quad i=\overline{1,n}.
\end{eqnarray}
Theorem \ref{myktina19} is proved.
\end{proof}

\begin{te}\label{MykHon1}
Let the matrix $||a_{ij}||_{i,j =1}^n$ be nonnegative, the vector $b=\{b_i\}_{=1}^n$ be strictly positive, and let the matrix $ ||a_{ij}||_{i,j \in I}$ be indecomposable for some nonempty set $I$. The necessary and sufficient condition of existence
of the equilibrium  price vector  $p=\{p_i\}_{i=1}^n$, which is a solution of the system of equations and inequalities
  $$ \sum\limits_{j=1}^n a_{ij}\frac{b_j p_j}{\sum\limits_{s=1}^n a_{sj}p_s}=b_i \quad i \in I, $$
\begin{eqnarray}\label{MykHon2}
   \sum\limits_{j=1}^n a_{ij}\frac{b_j p_j}{\sum\limits_{s=1}^n a_{sj}p_s}<b_i \quad i \in J,
\end{eqnarray}
in the set $P=\{p=\{p_i\}_{i=1}^n,\ p_i \geq 0,\ i=\overline{1,n}, \ \sum\limits_{i=1} ^n p_i=1\},$
such that $p_i>0, i \in I,$ is the existence of a strictly positive solution of the system of equations and inequalities
$$  \sum\limits_{j\in I}a_{ij} z_j=b_i \quad i \in I,     $$
\begin{eqnarray}\label{MykHon3}
   \sum\limits_{j\in I} a_{ij}z_j<b_i \quad i \in J. 
\end{eqnarray}
\end{te}
\begin{proof}
Necessity.
Let there exist a solution of the system of equations and inequalities (\ref{MykHon2}) with respect to the vector $p=\{p_i\}_{i=1}^n$ in the set $P$ and such that $p_i>0, i \in I.$ Whereas in the previous Theorem \ref{myktina19} we establish that $p_i=0, \ i \in J.$
The remaining component $p_i, \ i \in I,$ of the vector $p$ is the solution of the system of equations and inequalities
\begin{eqnarray}\label{MykHon4}
  \sum\limits_{j\in I} a_{ij}\frac{b_j p_j}{\sum\limits_{s\in I } a_{sj}p_s}=b_i \quad i \in I,
 \end{eqnarray}
 \begin{eqnarray}\label{MykHon5}
  \sum\limits_{j\in I} a_{ij}\frac{b_j p_j}{\sum\limits_{s\in I } a_{sj}p_s}<b_i \quad i \in J.
 \end{eqnarray}
Let's introduce the denotation
$$ z_i=\frac{b_i p_i}{\sum\limits_{s \in I}a_{si} p_s},\quad i \in I.$$
Then the equalities and inequalities 
 $$ \sum\limits_{j \in I} a_{ij} z_j=b_i \quad i \in I, $$
 $$ \sum\limits_{j \in I} a_{ij}z_j<b_i \quad i \in J. $$
 are valid.
It is obvious that $ z_i>0, i \in I.$ The necessity is established.

Sufficiency. If there exists a strictly positive solution of the system of equations and inequalities (\ref{MykHon3}), then the conditions of Theorem \ref{mykt19} are true, that is, there exists a strictly positive solution of the system of equations 
 \begin{eqnarray}\label{MykHon6}
\sum\limits_{i \in I}a_{ki} \frac{p_i b_i}{\sum\limits_{s \in I}a_{si}p_s} =b_k, \quad k\in I,
\end{eqnarray}
and inequalities
\begin{eqnarray}\label{MykHon7}
\sum\limits_{i \in I}a_{ki} \frac{p_i b_i}{\sum\limits_{s \in I}a_{si}p_s} < b_k, \quad k\in J.
\end{eqnarray}
 Let's construct the equilibrium vector of prices $p=\{p_i\}_{i=1}^n $ by setting $p_i=0, i \in J,$ and choosing the components $p_i, i \in I,$ to be equal to the corresponding components of solution of the system of equations and inequalities (\ref{MykHon6}),(\ref{MykHon7}). The price vector constructed in this way is the solution of the system of equations and inequalities (\ref{MykHon2}).
Theorem \ref{MykHon1} is proved.
\end{proof}
For any vector $z_0 \in Z$, let
\begin{eqnarray}\label{TVYA1}
\bar b=\sum\limits_{i=1}^n z_i^0 a_i.
\end{eqnarray}
If $I=\{i, \bar b_i=b_i\}$ is a nonempty set, then the vector $\bar b$
will be called the vector of real consumption. In accordance with Theorems \ref{1TtsyVtsja4} - \ref{TtsyVtsja5}, it corresponds to the equilibrium price vector $p_0,$ which is the solution of the system of equations (\ref{mari3}).

For a part of the goods, the indices of which are included in the set $J=N_0\setminus I,$ the equilibrium price is $p_i^0=0, \ i \in J.$ The latter means that this part of the goods does not find consumers on the market of goods of the economic system. But certain funds were spent on their production, which are called the cost of these goods.
Let's introduce the generalized equilibrium vector of prices by putting $p_u=\{p_i^u\}_{i=1}^n$ $p_i^u=p_i^0, \ i \in I,$ \ $ p_i^u=p_i^ c, i \in J,$ where $p_i^c$ is the cost price of the produced goods that do not find consumers on the market of goods of the economic system. Each such equilibrium state will be matched with the level of excess supply
\begin{eqnarray}\label{zaTVYA1}
R=\frac{\langle b-\bar b, p_u\rangle}{\langle b, p_u\rangle},
\end{eqnarray}
where $\langle x, y\rangle=\sum\limits_{i=1}^n x_i y_i, \ x=\{x_i\}_{i=1}^n, \ y=\{y_i\}_ {i=1}^n.$

Finding solutions of the system of equations and inequalities (\ref{mykt17}), (\ref{mykt18}) with the smallest excess supply  will require finding all possible solutions of such a system of equations and inequalities and finding among them the minimum excess supply, which can turn out to be an infeasible problem for large dimensions of the matrix $A.$ Based on Theorem \ref{TVYA} and Proposition \ref{g1} below, the solution of this problem is proposed as a quadratic programming problem.

\begin{de}\label{myktinavitka1}
Let $A$ be an indecomposable nonnegative matrix, and let $b$ be a strictly positive vector that does not belong to the cone formed by the column vectors of the matrix $A.$
The solution $z_0$ of the quadratic programming problem
\begin{eqnarray}\label{myktinavitka2}
\min\limits_{z \in Z}\sum\limits_{i=1}^n[b_i-\sum\limits_{k=1}^n a_{ik} z_k]^2,
\end{eqnarray}
where
$ Z=\{z=\{z_i\}_{i=1}^n, \ z_i \geq 0, i=\overline{1,n}, \sum\limits_{i=1}^n a_{ ki}z_i^0 \leq b_k, \ k=\overline{1,n}\}$
 corresponds to the real consumption vector $\bar b=A z_0\leq b.$
Assume that for the non-empty set $I=\{i, \bar b_i=b_i\}$ there exists an equilibrium price vector  $p_0=\{p_i^0\}_{i=1}^n$ which is a solution of the system equations
\begin{eqnarray}\label{20myktina21}
\sum\limits_{i=1}^n a_{ki} \frac{p_i^0\bar b_i}{\sum\limits_{s=1}^n a_{si}p_s^0}=\bar b_k, \quad k=\overline{1,n}.
\end{eqnarray}
 Then the value
\begin{eqnarray}\label{jajaja1}
  R=\frac{\langle b -\bar b, p_0 \rangle}{\langle b, p_0 \rangle }
\end{eqnarray}
   will be called the generalized excess supply  corresponding to the generalized equilibrium vector of prices $p_0,$ where
$ \langle x,y \rangle=\sum\limits_{i=1}^n x_i y_i,$
$x=\{x_i\}_{i=1}^n, \ y=\{y_i\}_{i=1}^n.$
\end{de}
According to the formula (\ref{jajaja1}), the level of excess supply  for the generalized equilibrium vector is the smallest. This is the state of economic equilibrium below which the economic system cannot fall.

\section{Economic systems capable of operating in sustainable  mode.}

A description of sustainable economic development at the macroeconomic level is proposed in \cite{6Gonchar}, \cite{2Gonchar}, \cite{1Gonchar}, \cite{8Gonchar}. In this section, we formulate the principles of sustainable development at the microeconomic level.
Each business project, which starts the production of a certain group of goods, plans to receive added value. In the production process, there are direct costs for production materials and labor costs.
If a business project is such that it has sales markets with positive added value, then we say that this business project has contracts with suppliers of materials and raw materials and a certain number of workers who sell their labor power. 
The totality of concluded contracts will be characterized by technological mapping. Consider a technological mapping that describes the production of the $i$-th type of product by spending a vector of goods
 $ \{a_{ki}\}_{k=1}^n, i=\overline{1,n},$ on one unit of the i-th manufactured product. Such a technological mapping will be characterized by the "input-output" matrix $|| a_{ki}||_{k,i=1}^n.$ Our goal is to formulate principles that will ensure sustainable development of the economic system.
 Suppose that the economic system produces $ x_i , \ i=\overline{1,n},$ units of the i-th product.
Then the gross input vector  is equal to $X_i= \{x_i a_{ki}\}_{k=1}^n$ and the gross output vector is equal to $Y_i= \{x_i \delta_{ki}\}_{k=1}^n.$
 Each $i-$th consumer in the economic system is characterized by two vectors
$b_i =\{b_{ki}\}_{k=1}^n, i=\overline{1,l}$ and $ C_i=\{c_{ki}\}_{k=1}^n , i=\overline{1,l}.$
 The vectors $b_i =\{b_{ki}\}_{k=1}^n, i=\overline{1,l},$ we call the property vectors.
Every $i$-th consumer wants to exchange the vector $b_i =\{b_{ki}\}_{k=1}^n$ for the vector $C_i=\{c_{ki}\}_{k= 1}^ n,$
which we call the demand vector.
We assume that in the production process all produced goods are distributed according to the rule
\begin{eqnarray}\label{pupvittin1}
\sum\limits_{i=1}^n (Y_i- X_i)=\sum\limits_{i=1}^l b_i.
\end{eqnarray}
Is there such a market mechanism that would provide
such distribution of the product in society for a certain vector of prices?

\begin{de}\label{pupvittin2}
The distribution of the product in society will be called economically expedient if, in the process of production and distribution, in accordance with the concluded agreement, the $i$th consumer owns a set of goods $b_i =\{b_{ki}\}_{k=1}^n, i=\overline{1,l},$ so that the vector $b=\sum\limits_{i=1}^l b_i$ belongs to the interior of the cone formed by the vectors
$ C_i=\{c_{ki}\}_{k=1}^n, i=\overline{1,l}.$
\end{de}

\vskip 5mm
Let's introduce two matrices $C=|c_{ki}|_{k=1, i=1}^{n,l}$ and $B=|b_{ki}|_{k=1, i=1 }^ {n,l}.$

\begin{de}\label{pupvittin2}
If the representation $B=CB_1$ is valid for the matrix $B$
and such that there is a solution to the problem
\begin{eqnarray}\label{pupvittin3}
\sum\limits_{k=1}^l b_{ki}^1d_k=\sum\limits_{s=1}^lb_{is}^1 d_i
\end{eqnarray}
with respect to the vector $d=\{d_k\}_{k=1}^l,$, which belongs to the cone formed by the vectors $C_i^T=\{c_{ki}\}_{k=1 }^n,$ then we will call the distribution of the product in society rational.
\end{de}

In this work, we adhere to the concept of describing economic systems developed in
\cite{Gonchar2}. The essence of this description is that the supply of firms is primary, and the choice of consumers is secondary. But at the same time, it is important that the structure of the supply corresponds to the structure of the demand.
The axioms of this description are presented in \cite{Gonchar2}, where random fields of consumer choice and decision-making by firms are constructed based on these axioms.

We describe firms by technological mappings $ y_i=F_i(x_i), x_i \in X_i $ from the CTM class (compact technological mapping)\cite{10Gonchar2},  \cite{11Gonchar2} \cite{Gonchar2}, and the demand of the $ i$ consumer by the product vector $C_i=\{c_{ki}\}_ {k=1}^n$, which he wants to consume in a certain period   of the economy functioning. Suppose that m firms function in the economic system, which are described by technological mappings $ y_i=F_i(x_i), x_i \in X_i, i=\overline{1,m}, $ from the CTM class.
 If firms have chosen production processes $ x_i \in X_i, y_i \in F(x_i), i=\overline{1,m},$ then the produced final product in the economic system will be equal to $ \sum\limits_{i=1}^m (y_i - x_i)$, and the $i-th$ consumer will receive a vector of goods $b_i, i=\overline{1,l}, $ in accordance with the signed contracts.
The condition of economic equilibrium is the fulfillment of inequalities

\begin{eqnarray}\label{pupvittin4}
\sum\limits_{i=1}^l c_{ki} \frac{ \langle b_i, p \rangle }{\langle C_i, p \rangle}\leq \sum\limits_{i=1}^l b_{ki}, 
\quad k=\overline{1,n}.
\end{eqnarray}
The main idea that we lay down is the following: firms must produce such a quantity of goods, the sale of which will ensure their functioning in the next production cycle, and for this they must have sufficient income.
In addition, the necessary balances for the state must be ensured: state spending on defense, ensuring the freedoms and rights of citizens, public administration, updating fixed assets, funding education, etc.

Is such a distribution of products the result of market exchange based on the existence of an equilibrium vector of prices. Among all the possible distributions of output produced by firms, what matters is when firms will be profitable so that they can undertake the next production cycle. 
Under the firm we also understand any owner of labor force that he sells on the market of the economic system for the appropriate salary, which he spends on its restoration and satisfaction of other needs.
In this case, it is convenient to consider all those engaged in the production of products and services as firms that produce labor and sell it to other firms.
We assume that the economic system has $l$ consumers, $m$ firms, $l>m$ and produces $n$ type of goods. In order for the economic system to function, it is necessary to ensure the protection of the rights and freedoms of citizens, public administration, protection from external threats and etc. For this, the income of firms should be taxed. Let us consider the taxation vector $\pi=\{\pi_i\}_{i=1}^m, \ 0 \leq  \pi_i<1, i=\overline{1,m}, $ whose economic meaning is as follows: $\pi_i \langle p, y_i -x_i \rangle$ is the part of the income that is withdrawn to finance education, defense, public administration and  other institutions for the safe functioning of the state. Here
$\langle p, y_i -x_i \rangle=\sum\limits_{k=1}^n p_k(y_{ki}-x_{ki}),$ where $x_i=\{x_{ki}\}_{k=1}^n,$ $ y_i=\{y_{ki}\}_{k=1}^n,$  are the input and output vectors and $p=\{p_k\}_{k=1}^n$ is an equilibrium price vector.

The general formulation of the problem is as follows:
firms implemented production processes $(x_i, y_i), \ x_i \in X_i, y_i \in F_i(x_i), i=\overline{1,m}$
 as a result of which the final product $ \sum\limits_{i=1}^m (y_i - x_i)$ is produced in the economic system, which must be distributed so that the process of production and distribution of products is continuous. A condition for this is a system of equalities
 \begin{eqnarray}\label{pupvittin5}
\sum\limits_{i=1}^m x_{ki} \frac{(1-\pi_i) \langle y_i, p \rangle }{\langle x_i, p \rangle}+\sum\limits_{i=m+1}^l\frac{C_{ki}D_i(p)}{ \langle C_i, p\rangle}=
\sum\limits_{i=1}^m y_{ki}, 
\quad k=\overline{1,n},
\end{eqnarray}

\begin{eqnarray}\label{tsytsjatin1}
\sum\limits_{i=m+1}^ly_i^0 C_{ki}= \sum\limits_{i=1}^m \pi_i y_{ki}, \quad k=\overline{1,n}.
\end{eqnarray}
Equalities (\ref{tsytsjatin1}) are  material balances to ensure public procurement, defense orders, construction of educational institutions, renewal of fixed assets, etc.
If we put $D_i(p)=y_i^0 \langle C_i, p\rangle$ and take into account (\ref{tsytsjatin1}), then
we get
 \begin{eqnarray}\label{tsytsjatin2}
\sum\limits_{i=1}^m x_{ki} \frac{(1-\pi_i) \langle y_i, p \rangle }{\langle x_i, p \rangle}=\sum\limits_{i=1}^m(1 - \pi_i) y_{ki},
\quad k=\overline{1,n},
\end{eqnarray}
 where we denoted
$$ \langle y_i, p \rangle =\sum\limits_{k=1}^n y_{ki}p_k, \quad \langle x_i, p \rangle =\sum\limits_{k=1}^n x_{ ki}p_k,$$
and $\pi=\{\pi_i\}_{I=1}^m$ is the taxation vector.

In order for the process of functioning of the economic system to be continuous, it is necessary that there should be an equilibrium vector of prices $p_0$, so that
equalities (\ref{pupvittin5}) and inequalities $\langle y_i -x_i, p_0 \rangle >0, i=\overline{1,m},$ were valid.

The ability of the economic system to function continuously will be called the ability to function in the mode of sustainable development.

Below we consider the "input - output" production model.
Suppose that $ x_i^0 $ units of the i-th product are produced in the economic system, $i=\overline{1,n}$.
In this case, the input vector  is equal to $X_i= \{x_i^0 a_{ki}\}_{k=1}^n.$
The output vector is equal to $Y_i= \{x_i^0 \delta_{ki}\}_{k=1}^n.$
Then the problem (\ref{tsytsjatin2}) can be written in the form
\begin{eqnarray}\label{tsytsjatin3}
\sum\limits_{i=1}^n a_{ki} \frac{(1-\pi_i)x_i^0 p_i}{\sum\limits_{s=1}^na_{si}p_s}=(1- \pi_k)x_k^0, \quad k=\overline{1,n}.
\end{eqnarray}
Denote $(1-\pi_i)x_i^0 =x_i, \ i=\overline{1,n},$ then the problem (\ref{tsytsjatin3}) and the conditions of profitability are rewritten in the form

\begin{eqnarray}\label{pupvittin11}
\sum\limits_{i=1}^n a_{ki} \frac{x_i p_i}{\sum\limits_{s=1}^na_{si}p_s}=x_k, \quad k=\overline{1, n},
\end{eqnarray}
\begin{eqnarray}\label{pupvittin12}
 p_i- \sum\limits_{s=1}^na_{si}p_s >0,\quad i=\overline{1,n}.
\end{eqnarray}

\begin{te}\label{1pupvittin13}
Let $ A=||a_{ki}||_{k,i=1}^n$ be a non negative productive indecomposable matrix.
The necessary and sufficient conditions for the  functioning of the economic system in the mode of sustainable development are that the vector $x=\{x_i\}_{i=1}^n$ belongs to the interior of the cone formed by the column vectors of the matrix $A(E-A)^{-1}.$
\end{te}
\begin{proof}
 Necessity. Assume that there exists a vector of equilibrium prices $p_0$ that satisfies the system of equations (\ref{pupvittin11}) and inequalities (\ref{pupvittin12}). Substituting $p_i^0$ from equalities
\begin{eqnarray}\label{pupvittin13}
 p_i^0- \sum\limits_{s=1}^na_{si}p_s^0 =\delta_i^0,\quad \delta_i^0>0,\quad i=\overline{1,n}.
\end{eqnarray}
into the system of equations (\ref{pupvittin11}) we get
\begin{eqnarray}\label{pupvittin14}
x_k =\sum\limits_{i=1}^n a_{ki} \frac{x_i( \sum\limits_{s=1}^na_{si}p_s^0 +\delta_i^0)}{\sum\limits_{s=1}^na_{si}p_s}=\sum\limits_{i=1}^n a_{ki} x_i +
\sum\limits_{i=1}^n a_{ki} \frac{x_i \delta_i^0}{\sum\limits_{s=1}^na_{si}p_s^0}, \quad k=\overline {1,n}.
\end{eqnarray}
If we introduce the vector $\alpha=\{\alpha_k\}_{k=1}^n$, where
$\alpha_k=\frac{x_k \delta_k^0}{\sum\limits_{s=1}^na_{sk}p_s^0}>0, \ k=\overline{1,n}, $ then we get from of equalities (\ref{pupvittin14}) the equality $x=A(E-A)^{-1}\alpha.$ The latter proves the necessity.
 
Sufficiency. From the very beginning, we assume that $ A^{-1}$ exists. For a diagonal matrix
\begin{eqnarray}\label{pupvittin15}
X=||\delta_{ij}x_j||_{i,j=1}^n
\end{eqnarray}
the representation $X=AB_1$ is true, where $B_1=A^{-1}X=||a_{ki}^{-1} x_i||_{k,i=1}^n.$
From the assumptions of Theorem \ref{1pupvittin13} we have
\begin{eqnarray}\label{pupvittin16}
b_k^1=\sum\limits_{i=1}^na_{ki}^{-1} x_i=[(E-A)^{-1}\alpha]_k>0, \ k=\overline{1,n }.
\end{eqnarray} 
 We will prove that the system of equations
\begin{eqnarray}\label{pupvittin17}
\sum\limits_{k=1}^na_{ki}^{-1} x_i d_k=b_i^1 d_i, \quad i=\overline{1,n},
\end{eqnarray}
has a strictly positive solution belonging to the cone formed by the vectors $\{a_{ki}\}_{k=1}^n, \ i=\overline{1,n}.$ From (\ref{pupvittin16}) we get
\begin{eqnarray}\label{pupvittin18}
x_k =\sum\limits_{i=1}^na_{ki} b_i^1.
\end{eqnarray}
The problem (\ref{pupvittin17}) is equivalent to the problem
\begin{eqnarray}\label{pupvittin19}
d_k=\sum\limits_{i=1}^na_{ik} \frac{b_i^1}{ x_i }d_i=\sum\limits_{i=1}^na_{ik} \frac{b_i^1}{ \sum\limits_{k=1}^na_{ik} b_k^1}d_i, \quad i=\overline{1,n}.
\end{eqnarray}
Let's introduce the denotation
\begin{eqnarray}\label{phuph0}
 u_{i k}=a_{ik}\frac{b_i^1}{ \sum\limits_{k=1}^na_{ik} b_k^1}, \quad i, k=\overline{1,n},
\end{eqnarray}
and the matrix $U=|u_{i k}|_{i,k=1}^n.$
 Consider a nonlinear system of equations
\begin{eqnarray}\label{pupvittin20}
d_k=\frac{d_k+\sum\limits_{i=1}^n u_{ik}d_i}{1+ \sum\limits_{k=1}^n \sum\limits_{i=1}^nu_{ik } d_i},
\quad k=\overline{1,n},
\end{eqnarray}
on the set $D=\{d=\{d_i\}_{i=1}^n, \ d_i\geq0, \ i=\overline{1,n}, \ \sum\limits_{i=1}^ n d_i=1\}.$

Thanks to Schauder's theorem \cite{Nirenberg}, there exists a solution of the system of equations (\ref{pupvittin20}) in the set $D.$ The system of equations (\ref{pupvittin20})
can be written as
\begin{eqnarray}\label{phuph1}
 \sum\limits_{i=1}^n u_{ik} d_i=\lambda d_k,
\end{eqnarray}
where $\lambda=\sum\limits_{k=1}^n \sum\limits_{i=1}^n u_{ik}d_i.$

We will prove that $\lambda>0$ and the solution $d=\{d_i\}_{i=1}^n$ of the system of equations is strictly positive due to the indecomposability of the matrix $A$.
Indeed, the vector $d=\{d_i\}_{i=1}^n$ satisfies the system of equations
(\ref{phuph1}), which can be written in operator form
\begin{eqnarray}\label{phuph3}
U^T d=\lambda d,
\end{eqnarray}
or
\begin{eqnarray}\label{phuph3}
[U^T]^{n-1} d=\lambda^{n-1} d.
\end{eqnarray}
Due to the fact that the vector $d$ belongs to the set $D,$ and the matrix $U$ is non-negative and indecomposable, the vector $[U^T]^{n-1} d$ is strictly positive. It follows that $\lambda>0$ and the vector $d$ is strictly positive.
Let's prove that $\lambda=1.$
The problem (\ref{phuph1}) is equivalent to the problem
\begin{eqnarray}\label{phuph2}
\lambda \sum\limits_{k=1}^na_{ki}^{-1} x_i d_k=b_i^1 d_i, \quad i=\overline{1,n}.
\end{eqnarray}
The summation by index $i$ of the left and right parts of equalities (\ref{phuph2}) gives $\lambda \sum\limits_{k=1}^n b_k^1 d_k=\sum\limits_{i=1}^nb_i^ 1 d_i.$
The latter proves necessary.
Therefore, there is a solution to the problem (\ref{pupvittin19}). From the Theorem \ref{100wickkteen3} and equalities (\ref{pupvittin19}), if we put
\begin{eqnarray}\label{pupvittin21}
p_i= \frac{b_i^1d_i}{ \sum\limits_{s=1}^na_{is} b_s^1}
\end{eqnarray}
and introduce the price vector $p=\{p_i\}_{i=1}^n$, then
$d_k=\sum\limits_{s=1}^n a_{sk}p_s, \quad k=\overline{1,n}.$

Or, in relation to the equilibrium price vector, we obtain a system of equations
\begin{eqnarray}\label{pupvittin22}
p_i= \frac{b_i^1}{ \sum\limits_{s=1}^na_{is} b_s^1} \sum\limits_{s=1}^n a_{si}p_s, \quad i=
\overline{1,n}.
\end{eqnarray}

Due to the fact that the representation $x=A(E-A)^{-1}\alpha$ is valid for the vector $x$,  where $\alpha$ is a strictly positive vector, then from (\ref{pupvittin18}) we get
$b^1=(E-A)^{-1}\alpha,$ where $b^1=\{b_k^1\}_{k=1}^n$. We have from the last one
\begin{eqnarray}\label{pupvittin23}
\frac{b_i^1}{ \sum\limits_{s=1}^na_{is} b_s^1}=1+\frac{\alpha_i}{\sum\limits_{k=1}^\infty\sum \limits_{s=1}^n a_{is}^k\alpha_s}>1, \quad i=\overline{1,n},
\end{eqnarray}
where we have denoted by $ a_{is}^k$ the matrix elements of the matrix $A^k.$
 Therefore, the constructed solution of the problem (\ref{pupvittin19}) is such that the inequalities 
\begin{eqnarray}\label{pupvittin24}
p_i- \sum\limits_{s=1}^n a_{si}p_s>0, \quad i=\overline{1,n},
\end{eqnarray}
hold, if the conditions of Theorem \ref{1pupvittin13} are valid.

Assume that the matrix $A$ is degenerate. The proof of necessity is the same as in the previous case. To prove sufficiency, consider a non-degenerate matrix $A+\varepsilon E,$ where $\varepsilon>0 $ and small enough such that $(A+\varepsilon E)^{-1}$ exists for all $\varepsilon>0. $ This is possible due to the fact that $\varepsilon=0$ is a root of the equation $\det(A+\varepsilon E)=0$ and this equation has a finite number of roots.
We assume that $x(\varepsilon)=(A+\varepsilon E) (E-A-\varepsilon E)^{-1}\alpha, $ where the vector $ \alpha$ is strictly positive. As before, repeating all the above arguments, we come to the fact that the vector $d(\varepsilon)=\{d_k(\varepsilon)\}_{k=1}^n$ satisfies the system of equations
\begin{eqnarray}\label{pupvittin25}
d_k(\varepsilon)=\sum\limits_{i=1}^n(a_{ik}+\delta_{ik}\varepsilon) \frac{b_i^1(\varepsilon)}{ \sum\limits_{k= 1}^n(a_{ik}+\delta_{ik}\varepsilon) b_k^1(\varepsilon)}d_i(\varepsilon), \quad i=\overline{1,n},
\end{eqnarray}
and it is strictly positive.

Due to the fact that $b^1(\varepsilon)=\{b_i^1(\varepsilon)\}_{i=1}^n=(E-A-\varepsilon E)^{-1}\alpha,$
$(A+\varepsilon E)b^1(\varepsilon)=(A+\varepsilon E)(E-A-\varepsilon E)^{-1}\alpha$ we get that there is a limit
$$\lim_{\varepsilon \to 0}b^1(\varepsilon)=(E-A)^{-1}\alpha=b^1=\{b^1_k\}_{k=1}^n. $$
 Because $d(\varepsilon)$ is bounded, there exists a subsequence $\varepsilon_n$ that goes to zero, so that there is a limit
$$\lim_{\varepsilon_n \to 0}d(\varepsilon_n)=d=\{d_k\}_{k=1}^n,$$
which satisfies the system of equations (\ref{pupvittin19}).

Following the above arguments, we come to the fact that the introduced vector $p=\{p_i\}_{i=1}^n$ satisfies the system of equations (\ref{pupvittin22}). Let us prove that the vector $x$ satisfies the system of equations (\ref{pupvittin11}). Really, the equality 
\begin{eqnarray}\label{pupvittin26}
\frac{x_i p_i}{\sum\limits_{s=1}^na_{si}p_s}=\frac{b_i^1x_i}{\sum\limits_{s=1}^na_{is}b^1_s} =b_i^1
\end{eqnarray}
is valid.

Multiplying the left and right parts by $a_{ki}$ and summing over $i,$ we get
\begin{eqnarray}\label{pupvittin27}
\sum\limits_{i=1}^n a_{ki} \frac{x_i p_i}{\sum\limits_{s=1}^na_{si}p_s}=\sum\limits_{i=1}^n a_{ki}b_i^1=x_k, \quad k=\overline{1,n}.
\end{eqnarray}
The remaining statements of Theorem \ref{1pupvittin13} are obtained as before. The theorem \ref{1pupvittin13} is proved.
\end{proof}
 
The following Theorems are central in this section.
\begin{te}\label{tsytsjatintsytsja1}
Let $A=||a_{ki}||_{k,i=1}$  be an indecomposable and productive matrix for the  "input - output " production model. Suppose that  the strictly positive gross output vector
$x=\{x_i\}_{i=1}^n$ satisfies the system of equations
\begin{eqnarray}\label{tax1}
x_k-\sum\limits_{i=1}^na_{ki} x_i=c_k+e_k - i_k, \quad k=\overline{1,n},   
\end{eqnarray}
and the price vector 
$p=\{p_i\}_{i=1}^n$ satisfies the system of equations
\begin{eqnarray}\label{tax2}
p_ i-\sum\limits_{s=1}^na_{si} p_s  =\delta_i,\quad i=\overline{1,n},
\end{eqnarray}
where $c=\{c_i\}_{i=1}^n,$ $e=\{e_i\}_{i=1}^n,$ 
$i=\{i_k\}_{k=1}^n,$ 
$\delta=\{\delta_i\}_{i=1}^n, $ 
$c_k \geq 0,$ $ e_k \geq 0,$ $i_k \geq 0, $ $\delta_k >0,$ $k=\overline{1,n},$ are vectors of final consumption, export, import and added values, respectively. Then there exists a vector of taxation 
$\pi=\{\pi_i\}_{i=1}^n$ such that the vector 
$p=\{p_i\}_{i=1}^n$
 satisfies also the system of equations
\begin{eqnarray}\label{tsytsjatintsytsja2}
\sum\limits_{i=1}^n a_{ki} \frac{(1-\pi_i)x_i p_i}{\sum\limits_{s=1}^na_{si}p_s}= (1-\pi_k)x_k, \quad k=\overline{1,n}.
\end{eqnarray}
\end{te}
\begin{proof}  
Since the price vector $p=\{p_i\}_{i=1}^n$ is a strictly positive one, as a solution to the set of equations (\ref{tax2}), let's  denote $X_i=x_i p_i,$ $\Delta_i=\delta_i x_i,$ $\bar a_{ki}=\frac{p_k a_{ki}}{p_i}.$ 
In these denotations, the system of equations  (\ref{tsytsjatintsytsja2}) is written in the form
\begin{eqnarray}\label{2gonpupvittin2}
\sum\limits_{i=1}^n \bar a_{ki} \frac{(1-\pi_i)X_i }{\sum\limits_{s=1}^n\bar a_{si}}= (1-\pi_k)X_k, \quad k=\overline{1,n}.
\end{eqnarray}
Let us put
$$ V_i=\frac{X_i(1-\pi_i)}{\sum\limits_{s=1}^n\bar a_{si}},  \quad i=\overline{1,n}.$$
Then the vector  $V=\{V_i\}_{i=1}^n$ satisfies the system of equations
\begin{eqnarray}\label{tax3}
 \sum\limits_{i=1}^n \bar a_{ki}V_i=\sum\limits_{s=1}^n\bar a_{sk}V_k,  \quad k=\overline{1,n}.  
\end{eqnarray}
On the basis of Lemma  \ref{pupvittin7}, there exists a strictly positive solution of this system of equations, which is determined with accuracy up to a constant. Let us 
denote $V_0=\{V_i^0\}_{i=1}^n$ 
the solution of this system of equations, the sum of the components of which is equal to one. Then 
$V=\{V_i\}_{i=1}^n=c_0V_0=\{c_0 V_i^0\}_{i=1}^n.$
From here
$$1-\pi_i=\frac{c_0V_i^0\sum\limits_{s=1}^n\bar a_{sk}}{X_i} =c_0\frac{V_i^0}{X_i}\left(1-\frac{\Delta_i}{X_i}\right),\quad i=\overline{1,n},$$
or
\begin{eqnarray}\label{tax4}
\pi_i=1-c_0\frac{V_i^0}{X_i}\left(1-\frac{\Delta_i}{X_i}\right), \quad i=\overline{1,n}. 
\end{eqnarray}
The constant $c_0>0$ can be chosen such that the inequalities 
$$ 1-c_0\frac{V_i^0}{X_i}\left(1-\frac{\Delta_i}{X_i}\right)>0, \quad i=\overline{1,n}, $$
are satisfied. 
Theorem \ref{tsytsjatintsytsja1} is proved.
\end{proof}
\begin{ce}\label{tinaKmoja1}
The best system of taxation $\pi=\{\pi_i\}_{i=1}^n$ under the condition that the final product will be created in the economic system is one that satisfies the equality
\begin{eqnarray}\label{tinaKmoja2}
\frac{1-\pi_i}{\frac{V_i^0}{X_i}(1-\frac{\Delta_i}{X_i})}=\frac{1}{\max\limits_{1\leq i\leq n}\frac{V_i^0}{X_i}\left(1-\frac{ \Delta_i}{X_i}\right)}, \quad i=\overline{1,n}.
\end{eqnarray}
\end{ce}
\begin{proof}
We choose the constant $c_0$ in the formula (\ref{tax4}) so that the value of 
$\pi_i, \ i=\overline{1,n}, $ is the smallest. For this, we should put
\begin{eqnarray}\label{tinaKmoja3}
c_0 =\frac{1}{\max\limits_{1\leq i\leq n}\frac{V_i^0}{X_i}\left(1-\frac{\Delta_i}{X_i}\right)}.
\end{eqnarray}
Then equalities (\ref{tinaKmoja2}) will hold.
\end{proof}

In the following Theorem \ref{main1}, we find out for the taxation system 
$\pi=\{\pi_i\}_{i=1}^n,$  the conditions  under which in the economic system  a final product can be created  and the economy  can function in the mode of sustainable development.

\begin{te}\label{main1}
 Let the matrix $ A=||a_{ki}||_{k,i=1}^n$ be non-decomposable and productive for the  "input - output" economy model. Then for the equilibrium price vector $p=\{p_i\}_{i=1}^n,$ which is a solution of the system of equations (\ref{tax2}), there always exists  a solution of the system of equations (\ref{tsytsjatintsytsja2})  with respect to the vector $x=\{x_i\}_{i=1}^n,$ which satisfies the system of equations (\ref{tax1})
with a strictly positive right-hand side of this system of equations, provided that the taxation system
$\pi=\{ \pi_i\}_{i=1}^n,$ satisfies the conditions
\begin{eqnarray}\label{main2}
 0 < \pi_i\leq 1- b\left(1-\frac{\Delta_i}{X_i}\right), \quad i=\overline{1,n},
 \end{eqnarray}
 for a certain  $b,$ satisfying inequalities
 $$ \max\limits_{1 \leq i \leq n}(1- \pi_i)<b < \frac{1}{ \max\limits_{1 \leq i \leq n}\left(1-\frac{\Delta_i}{X_i}\right)},  $$
where as before $\Delta_k=\delta_k x_k, X_k =p_k x_k, \ k=\overline{1,n}$
 \end{te}
 \begin{proof}
Here and further we use the denotations of Theorem \ref{tsytsjatintsytsja1}. Let the economic system be described by the "input-output" model, where the output vector satisfies the system of equations (\ref{tsytsjatintsytsja2}), and the equilibrium price vector is determined by the system of equations (\ref{tax2}).

Let's put $X_k=x_kp_k, $ $\bar A=|\bar a_{ki}|_{k, i=1}^n, \bar a_{ki}=\frac {p_k a_{ki}}{p_i} ,$
 $C_k=c_kp_k,$
$E_k=e_kp_k,$ $I_k=i_kp_k,$ $\Delta_k=\delta_kp_k.$
In these denotations, the system of equations (\ref{tsytsjatintsytsja2}) can be written in the form
\begin{eqnarray}\label{main3}
\sum\limits_{i=1}^n \bar a_{ki} \frac{(1-\pi_i)X_i }{\sum\limits_{s=1}^n\bar a_{si}}= (1 -\pi_k)X_k, \quad k=\overline{1,n}.
\end{eqnarray}
and the system of equations (\ref{tax2}) is in the form
\begin{eqnarray}\label{main7}
X_i-\sum\limits_{s=1}^n \bar a_{si} X_i=\Delta_i, \quad i=\overline{1,n}.
\end{eqnarray}
Introducing the denotation
\begin{eqnarray}\label{main4}
\frac{(1-\pi_i)X_i }{\sum\limits_{s=1}^n\bar a_{si}}=X_i^0, \quad i=\overline{1,n},
\end{eqnarray}
 the system of equations (\ref{main3}) can be rewritten in the form
 \begin{eqnarray}\label{main5}
\sum\limits_{i=1}^n \bar a_{ki} X_i^0 = \sum\limits_{s=1}^n\bar a_{sk} X_k^0, \quad k=\overline{1 ,n}.
\end{eqnarray}
On the basis of Lemma \ref{pupvittin7}, there is a strictly positive solution of this system of equations, which is determined with accuracy up to a constant.
Let $X_0=\{X_i^0\}_{i=1}^n$ be a strictly positive solution of the system of equations (\ref{main5}) whose component sum is equal to one. Then any other is equal to
$c_0X_0=\{c_0X_i^0\}_{i=1}^n$, where $c_0>0.$
For the components of the vector $X=\{X_i\}_{i=1}^n$
we get the formulas
 \begin{eqnarray}\label{main6}
 X_i =c_0 \frac{\sum\limits_{s=1}^n\bar a_{si}}{(1-\pi_i)} X_i^0=c_0\frac{\left(1-\frac{\Delta_i}{X_i} \right)}{(1-\pi_i)}X_i^0=c_0\frac{\left(1-\frac{\delta_i}{p_i}\right)}{(1-\pi_i)}X_i^0, \quad i=\overline{1,n }.
\end{eqnarray}
To prove that the obtained solution is a vector of gross output for some vector of final consumption, it is necessary to establish that the inequalities
\begin{eqnarray}\label{main7}
X_k-\sum\limits_{i=1}^n \bar a_{ki} X_i>0, \quad k=\overline{1,n},
\end{eqnarray}
are valid.
Because the inequalities (\ref{main2}) hold, the inequalities 
\begin{eqnarray}\label{main8}
1 \leq \frac{1- 
\pi_i}{b(1-\frac{\Delta_i}{X_i})}=\frac{1- 
\pi_i}{b\sum\limits_{s=1}^n \bar a_{si} }
\end{eqnarray}
are true,
therefore
$$X_k-\sum\limits_{i=1}^n \bar a_{ki} X_i \geq X_k-\frac{1}{b}\sum\limits_{i=1}^n \bar a_{ki}\frac{(1- 
\pi_i) X_i}{\sum\limits_{s=1}^n \bar a_{si} }=$$
\begin{eqnarray}\label{main9}
\left(1-\frac{1-\pi_k}{b} \right)X_k>0, \quad k=\overline{1,n}.
\end{eqnarray}
So,
\begin{eqnarray}\label{main10}
X_k-\sum\limits_{i=1}^n \bar a_{ki} X_i=Y_k, \quad k=\overline{1,n},
\end{eqnarray}
where the vector $Y=\{Y_k\}_{k=1}^n$ has strictly positive components.
Introducing the denotations $Y=C+E-I, $
 $C=\{C_i\}_{i=1}^n,$ $E=\{E_i\}_{i=1}^n,$ $I=\{I_i\}_{i=1} ^n,$
we get that the vector $X=\{X_i\}_{i=1}^n$ is the gross output vector for the tax system under consideration. The proof of Theorem \ref{main1} ends with a remark about the problem (\ref{tsytsjatintsytsja2}) and (\ref{main3}) are equivalent due to the strict positivity of the price equilibrium vector.
\end{proof}

The following Theorem \ref{2gonpupvittin3} considers the case of a taxation system under which the economic system can still function in the mode of sustainable development. This case is interesting because the added value created in the relevant industry is equal to the value of the final product created in the same industry in a state of economic equilibrium.

\begin{te}\label{2gonpupvittin3}
Let the matrix $\bar A=||\bar a_{ki}||_{k,i=1}^n$ be indecomposable, whose matrix elements satisfy the system of inequalities
\begin{eqnarray}\label{main11}
 \sum\limits_{s=1}^n \bar a_{si}<1, \quad i=\overline{1,n}.
\end{eqnarray}
Then for taxation systems $\pi=\{\pi_i\}_{i=1}^n,$ where $ \pi_i= \frac{\delta_i}{p_i}=\frac{\Delta_i}{X_i}, i =\overline{1,n},$ the economic system is able to function in the mode of sustainable development.
The gross output vector $X=\{X_i\}_{i=1}^n$ satisfies the system of equations
\begin{eqnarray}\label{2gonpupvittin1}
\sum\limits_{i=1}^n \bar a_{ki} X_i=\sum\limits_{s=1}^n \bar a_{sk}X_k, \quad k=\overline{1,n},
\end{eqnarray}
the solution of which exists and is
strictly positive.
The vector of final consumption in value indicators is given by the formula $Y=\{Y_i\}_{i=1}^n,$
$Y_i=d (1 - \sum\limits_{s=1}^n \bar a_{si}) X_i^0, \ i=\overline{1,n},$ $d>0$ and is arbitrary, where $ X_0=\{X_i^0\}_{i=1}^n$ is one of the possible solutions of the system of equations (\ref{2gonpupvittin1}), the sum of whose components is equal to one.
The gross added value $\Delta_i^0$ of the $i$-th industry is given by the formula
$\Delta_i^0=d (1 - \sum\limits_{s=1}^n \bar a_{si}) X_i^0, \ i=\overline{1,n}.$
\end{te}
\begin{proof}
Due to the fact that $(1-\pi_i)=\sum\limits_{s=1}^n\bar a_{si},$ the system of equalities
(\ref{main3}) turns into a system of equations (\ref{2gonpupvittin1}) with respect to the gross output vector $ X=\{X_i\}_{i=1}^n$ in value indicators.
Based on Lemma \ref{pupvittin7}, there exists a strictly positive solution of the system of equations (\ref{2gonpupvittin1}). Let us denote one of the possible solutions $X_0=\{X_i^0\}_{i=1}^n,$
the sum of whose components is equal to one. We will rewrite the system of equations (\ref{2gonpupvittin1}) in the form
\begin{eqnarray}\label{2gonpupvittin4}
d X_k^0-\sum\limits_{i=1}^n \bar a_{ki}d X_i^0=(1-\sum\limits_{s=1}^n \bar a_{sk})d X_k ^0, \quad k=\overline{1,n}.
\end{eqnarray}

The gross output vector $d X_0=\{dX_i^0\}_{i=1}^n$ is uniquely determined from the system of equations
 (\ref{2gonpupvittin4}). Therefore, the final consumption vector is given by the formula
 $Y=\{Y_i\}_{i=1}^n,$
$Y_i=d (1 - \sum\limits_{s=1}^n \bar a_{si}) X_i^0, \ i=\overline{1,n},$ $d>0.$
The created gross added value is given by the formulas $\Delta_i^0=d (1 - \sum\limits_{s=1}^n \bar a_{si}) X_i^0, \ i=\overline{1,n}.$
The theorem \ref{2gonpupvittin3} is proved.
\end{proof}

In the future, we will approach the models of real economic systems with models of economic systems capable of functioning in the mode of sustainable development.
One of such algorithms is presented in the following corollary.
\begin{ce}\label{2gonpupvittin5} Let $Y=\{Y_i\}_{i=1}^n$ be the vector of the final product created in the real economic system. Then, in the mean square, the closest to the vector of the created final product in the real economic system is the vector of final consumption $Y_0=\{Y_i^0\}_{i=1}^n,$ in the economic system capable of functioning in the mode of sustainable development  under the consideration of Theorem \ref{2gonpupvittin3} taxation system, where
$$ Y_i^0=d_0(1-\sum\limits_{s=1}^n \bar a_{si})X_i^0, \quad i=\overline{1,n}, \quad d_0=\frac {\sum\limits_{i=1}^nY_i \left(1-\sum\limits_{s=1}^n \bar a_{si}\right)X_i^0}{\sum\limits_{i=1}^n\left[ \left(1-\sum\limits_{s=1}^n \bar a_{si}\right)X_i^0\right]^2}.$$
The set of gross added values has the form
\begin{eqnarray}\label{2gonpupvittin6}
\Delta_i^0=d_0 \left(1 - \sum\limits_{s=1}^n \bar a_{si}\right) X_i^0, \ i=\overline{1,n}.
\end{eqnarray}
\end{ce}

\section{Research of real economic systems.}

Information about the country's economic system is given in value indicators, the structure of which is given by a table that can be characterized by the matrix of direct costs $\bar A= ||\bar a_{ki}||_{k,i=1}^n,$ the vector of gross outputs $X=\{X_i\}_{i=1}^n,$
the vector of gross added values
$\Delta=\{\Delta_k\}_{k=1}^n,$  the final consumption vector $C=\{C_i\}_{i=1}^n,$
the export vector $E=\{E_i\}_{i=1}^n,$ the import vector $I=\{I_i\}_{i=1}^n,$ between which there are relations
\begin{eqnarray}\label{1main1}
X_k-\sum\limits_{i=1}^n \bar a_{ki}X_i=C_k+E_k- I_k, \quad k=\overline{1,n}.
\end{eqnarray}
Gross added  value consists of taxes on production, which we denote by $T_1=\{T_i^1\}_{i=1}^n$ and wages and profits and mixed income, which we denote by $Z_1=\{Z_i^1\}_{i=1}^n.$
Let's put $\pi_0=\{\pi_i^0 \}_{i=1}^n,$ where the denotation $\pi_i^0=\frac{T_i^1}{\Delta_i} $ is introduced.
The vector $\pi_0=\{\pi_i^0\}_{i=1}^n$ is called the vector of taxation in the real economic system.
Then the conditions of sustainable development can be written in the form
\begin{eqnarray}\label{tax8}
\sum\limits_{i=1}^n \bar a_{ki} \frac{(1-\pi_i^0)X_i}{\sum\limits_{s=1}^n \bar a_{si}} = (1-\pi_k^0)X_k, \quad k=\overline{1,n},
\end{eqnarray}
\begin{eqnarray}\label{tax9}
 1 - \sum\limits_{s=1}^n \bar a_{si} >0,\quad i=\overline{1,n}.
\end{eqnarray}
But these conditions may not be fulfilled for real economic systems.

To study real economic systems, the statistical indicators of which are given in value indicators, it is convenient to introduce the concept of a relative equilibrium vector of prices. If the price vector $p=\{p_i\}_{i=1}^n,$ implemented in the economic system is not in equilibrium, then by introducing the relative price vector $\hat p=\{\hat p_i\}_{ i=1}^n, $ where $\hat p_i=\frac{p_i^0}{p_i},$ and $p_0=\{p_i^0\}_{i=1}^n$ is the equilibrium price vector , the condition of economic equilibrium in value indicators with respect to the vector of relative equilibrium prices $\hat p=\{\hat p_i\}_{i=1}^n$ can be written in the form

\begin{eqnarray}\label{3gonpupvittin1}
\sum\limits_{i=1}^n \bar a_{ki} \frac{(1-\pi_i^0)\hat p_i X_i}{\sum\limits_{s=1}^n   \hat p_s \bar a_{si}} \leq (1-\pi_k^0)X_k, \quad k=\overline{1,n},
\end{eqnarray}

In the future, we will study the system of inequalities (\ref{3gonpupvittin1}) and find out the conditions for the existence of a relative equilibrium vector of prices $\hat p_i=\frac{p_i^0}{p_i},$ we will partially describe the set of equilibrium states and find out the quality of each equilibrium state , based on proven Theorems \ref{mykt6}, \ref{mykt19}, \ref{mykt27} and Corollaries \ref{mykt14}, \ref{mykt16} with Theorems \ref{mykt6}, \ref{mykt19}.
First, let's find out when the conditions of sustainable development are fulfilled, that is, when there is equality in a system of inequalities (\ref{3gonpupvittin1}).
 Let $\hat p=\{\hat p_i\}_{i=1}^n$ be a strictly positive relative equilibrium price vector. We will find out the conditions when the conditions of sustainable development will be fulfilled
\begin{eqnarray}\label{3main1}
\sum\limits_{i=1}^n \bar a_{ki} \frac{(1-\pi_i^0)\hat p_i X_i}{\sum\limits_{s=1}^n \hat p_s \bar a_{si}} = (1-\pi_k^0)X_k, \quad k=\overline{1,n},
\end{eqnarray}
\begin{eqnarray}\label{3main2}
 \hat p_i - \sum\limits_{s=1}^n \hat p_s \bar a_{si}=\hat \delta_i, ,\quad i=\overline{1,n}.
\end{eqnarray}
where $\hat \delta_i>0, i=\overline{1,n}. $

An analogue of Theorem \ref{main1}
in the case under consideration there is

\begin{te}\label{2main1}
 Let the matrix $ \hat A=||\hat a_{ki}||_{k,i=1}^n$ be indecomposable and productive in the   "input - output " economy model, given in value indicators. Then for the relative equilibrium strictly positive price vector $\hat p=\{\hat p_i\}_{i=1}^n,$ which satisfies the system of equations (\ref{3main2})
 there is always a solution to the system of equations (\ref{3main1})
 with respect to the vector $X=\{X_i\}_{i=1}^n,$ for which the final product will be created, i.e. the inequalities
 $$ X_k-\sum\limits_{i=1}^n \bar a_{ki} X_i=Y_k>0, \quad k=\overline{1,n},$$
 are valid, provided that the taxation system
$\pi_0=\{ \pi_i^0\}_{i=1}^n$ satisfies the conditions
\begin{eqnarray}\label{3main3}
 0 < \pi_i^0\leq 1- b\left(1-\frac{\hat \Delta_i}{\hat X_i}\right), \quad i=\overline{1,n},
\end{eqnarray}
for a certain $b,$ satisfying inequalities
$$ \max\limits_{1 \leq i \leq n}(1- \pi_i^0)<b < \frac{1}{ \max\limits_{1 \leq i \leq n}\left(1-\frac{\hat \Delta_i}{\hat X_i}\right)},  $$
where $ \hat a_{ki}=\frac{\hat p_k \bar a_{ki}}{\hat p_i}, $ $\hat \Delta_k=\hat \delta_k X_k, \hat X_k =\hat p_k X_k, \ k=\overline{1,n}.$
Under these conditions, the economic system is able to function in the mode of sustainable development.
 \end{te}
The most interesting for applications is the case when the relative equilibrium vector of prices $\hat p=\{\hat p_i\}_{i=1}^n,$ is such that $\hat p_i=1, i=\overline{1,n }.$ Then $\Delta_i=X_i(1-\sum\limits_{s=1}^n \bar a_{si})=\hat \delta_i X_i, \ i=\overline{1,n }.$
 We reformulate this theorem for this case. 

\begin{te}\label{4main1}
 Let the matrix $ \hat A=||\hat a_{ki}||_{k,i=1}^n$ be indecomposable and productive in the  "input - output" economy model, given in value indicators. Then for the relative equilibrium strictly positive price vector $\hat p=\{\hat p_i\}_{i=1}^n, $ where $\hat p_i =1,\ i=\overline{1,n},$ which satisfies the system of equations (\ref{3main2})
 there exists  always a solution to the system of equations (\ref{3main1})
 with respect to the vector $X=\{X_i\}_{i=1}^n,$ for which the final product will be created, i.e. the inequalities
 $$ X_k-\sum\limits_{i=1}^n \bar a_{ki} X_i=Y_k>0, \quad k=\overline{1,n},$$
 are valid,  provided that the taxation system
$\pi_0=\{ \pi_i^0\}_{i=1}^n$ satisfies the conditions
\begin{eqnarray}\label{5main5}
 0 < \pi_i^0\leq 1- b\left(1-\frac{ \Delta_i}{X_i}\right), \quad i=\overline{1,n},
\end{eqnarray}
for a certain $b,$ satisfying inequalities
$$ \max\limits_{1 \leq i \leq n}(1- \pi_i^0)<b < \frac{1}{ \max\limits_{1 \leq i \leq n}\left(1-\frac{ \Delta_i}{X_i}\right)}.  $$
If, in addition $Y_k=C_k+E_k- I_k, \ k=\overline{1,n},$
then, under these conditions, the economic system, described in value indicators, is able to function in the mode of sustainable development.
 \end{te}
 
 If the inequalities (\ref{3gonpupvittin1}) hold, then there exists a nonempty subset $I \subseteq N,$ where $N=\{1,2,\ldots,n\},$ such that for indices $ k \in I $ inequalities (\ref{3gonpupvittin1}) are transformed into equalities. Under the condition $I=N$, we are talking about a complete clearing of the markets.
If $I\subset N$, we are talking about partial clearing of the markets. We call the vector $\bar \psi=\{\bar \psi_k\}_{k=1}^n$ the real consumption vector, where

\begin{eqnarray}\label{pupapups11}
\bar \psi_k=\sum\limits_{i \in I} \bar a_{ki} \frac{(1-\pi_i^0) \hat p_i X_i^0}{\sum\limits_{s \in I} \hat p_s \bar a_{si}}, \quad k=\overline{1,n},
\end{eqnarray}

In the case of partial market clearing, the real consumption vector $\bar \psi $ does not coincide with the supply vector $\psi=\{\psi_k\}_{k=1}^n,$ where $\psi_k=(1-\pi_k^ 0)\hat X_k^0.$ The vector of real consumption can be presented in the form
\begin{eqnarray}\label{pupapups12}
\bar \psi_k=\sum\limits_{i\in I } \bar a_{ki} y_i, \quad k=\overline{1,n},
\end{eqnarray}
 where the non-zero vector $y=\{y_i\}_{i \in I}$ is the solution of the system of equations and inequalities
\begin{eqnarray}\label{pupapups13}
\sum\limits_{i \in I } \bar a_{ki} y_i= \psi_k, \quad k \in I,
\end{eqnarray}
\begin{eqnarray}\label{pupapups14}
\sum\limits_{i \in I } \bar a_{ki} y_i< \psi_k, \quad k \in N \setminus I,
\end{eqnarray}
and the equilibrium vector of prices $ \hat p=\{\hat p_i\}_{i \in I}$ satisfies the system of equations 
\begin{eqnarray}\label{pupapups15}
y_i= \frac{(1-\pi_i^0) \hat p_i X_i^0}{\sum\limits_{s \in I} \hat p_s \bar a_{si}}, \quad i \in I.
\end{eqnarray}

According to Theorem \ref{mykt19} and Corollary \ref{mykt16}, there exists one to one   correspondence between the solutions of the system of equations and inequalities (\ref{pupapups13}), (\ref{pupapups14}) and the system of equations  (\ref{pupapups15}) with respect to the equilibrium price vector $ \hat p=\{\hat p_i\}_{i \in I}.$
Let us consider the relative generalized equilibrium vector of prices
$p_u=\{p_i^u\}_{i=1}^n, $ $p_i^u= \hat p_i, \ i \in I,$ $p_i^u= 1, \ i \in J,$
where the vector $\hat p^I=\{ \hat p_i\}_{i \in I}$ satisfies the system of equations (\ref{pupapups15}).

The number of goods, the cost of which is $\langle \psi-\bar \psi, p_u \rangle$, does not find a consumer on the market of goods of the economic system. We will characterize this phenomenon as the level of excess supply
\begin{eqnarray}\label{yavitatina}
R=\frac{\langle \psi-\bar \psi, p_u \rangle}{\langle \psi, p_u \rangle}.
\end{eqnarray}
If there is no state of economic equilibrium corresponding to the vector $y=\{y_i\}_{i=1}^n$, which satisfies the system of inequalities and equations

\begin{eqnarray}\label{Npupapups13}
\sum\limits_{i=1}^n \bar a_{ki} y_i= \psi_k, \quad k \in I,
\end{eqnarray}
\begin{eqnarray}\label{Npupapups14}
\sum\limits_{i=1}^n \bar a_{ki} y_i< \psi_k, \quad k \in N \setminus I,
\end{eqnarray}
with a non-empty set $J,$ as it happened in Theorem \ref{TtsyVtsja3}, then in this case 
the generalized relative equilibrium price vector $ \hat p=\{\hat p_k\}_{k=1}^n,$  is a  solution of the system of equations
\begin{eqnarray}\label{Npupapups15}
y_i=\frac{\bar \psi_i\hat p_i}{\sum\limits_{s=1}^n \bar a_{si} \hat p_s}, \quad i=\overline{1,n}.
\end{eqnarray}
where the vector of real consumption $\bar \psi $ is given by the formula
$\bar \psi=\{ \bar \psi_k\}_{k=1}^n,$ $\psi_k=\sum\limits_{i=1}^n \bar a_{ki} y_i, \ k= \overline{1,n}. $
and the vector $y=\{y_i\}_{i=1}^n$ is a solution of the set of equalities (\ref{Npupapups13}) and inequalities (\ref{Npupapups14}).
In this case,  the level of excess supply for such an equilibrium state will be described by the formula
\begin{eqnarray}\label{Npupapups16}
R=\frac{\langle \psi-\bar \psi, \hat p \rangle}{\langle \psi, \hat p \rangle}.
\end{eqnarray}

\section{Axioms of the  aggregate description of economy.}
Let's assume that $ m$ types of goods are produced in the economic system. We assume that all produced goods are ordered, and any set of them will be denoted by a  non-negative vector of  $ m$-dimensional subset  $R_+^m.$
In reality, the production of any product can be represented by a production process: $(x,y) $ where the vector $x $ is an input vector, and $y $ is an output vector. 
\begin{ax}\label{votre1}
Production of goods and services in the economic system  can be presented as a set of elementary production processes  $(x^i,y^i), \ i=\overline{1,m},$ to produce the $i$-th type of good $y^i=\{\delta_{ij}x_i\}_{j=1}^m$ with a certain vector of input $x^i,$ where $\delta_{ij}$ is a Kronecker symbol.
\end{ax}

\begin{ax}\label{votre2}
There is a linear relationship between vectors $x^i$ and $ y^i$ of productive process $(x^i,y^i) $ : to produce a vector of goods $y^i=\{\delta_{ij}x_i\}_{j=1}^m$ it is necessary to spend a vector of goods $ x^i=\{a_{ij}x_i\}_{j=1}^m,\ m=\overline{1,m},$ where $a_{ij}\geq 0, i, j=\overline{1,m}.$
\end{ax}
As a result of the assumption of linearity of elementary production processes, we obtain:
if to introduce into consideration the matrix $A=|a_{ij}|_{i,j=1}^m$ of direct costs, 
 then to produce the output vector $ y=\{x_i\}_{i=1}^m,$ it needs to  expend the  vector $x=Ay. $
The vector $y-x=y- Ay=c$ is called the final consumption vector. Consumer needs are considered primary, that is, the supply is primary, and the demand is secondary, therefore, as a rule, the final consumption vector $c=\{c_i\}_{i=1}^m$ is important.
If we introduce a re-designation, denoting the output vector by $x=\{x_i\}_{i=1}^m$ and the input vector by $Ax,$ then the main output equation  takes the form
\begin{eqnarray}\label{agr1}
x-Ax=c.
\end{eqnarray}
\begin{de}\label{agre1}
An non negative matrix $A$ is called productive if at least for one strictly positive vector $c_0 \in R^m$ there exists a strictly positive solution $x_0 $ of the system of equations (\ref{agr1}).
\end{de}
\begin{ax}\label{votre3}
The matrix of elementary direct costs $A$ is a productive one.
\end{ax}

\begin{prope}\label{agre2}
If the matrix $A$ is productive, then its spectral radius is strictly less than unity.
\end{prope}
\begin{proof}
For the  vector  $x \in R^m$, we introduce a norm by setting $||x||=\max\limits_{k \in M}\frac{|x_k|}{x_k^0},$ where $x^ 0=\{x_i^0\}_{i=1}^m$ is a solution of the system of equations (\ref{agr1})
for a certain strictly positive vector $c^0=\{c_i^0\}_{i=1}^m \in R^m.$
Then for
norm of the operator $A$ in the space $R^m$
there is the estimate $||A||\leq (1-\max\limits_{k \in M}\frac{c_k^0}{x_k^0})<1.$ It immediately follows that the system of equations ( \ref{agr1}) is solvable for any non-negative vector $c \in R^m$ in the set of non-negative solutions.
\end{proof}

 Now, to present the information
about the economic system in its entirety is to provide information about the $A$ matrix itself.
But more than $10^9$ types of goods are produced in the economic system. Therefore, such a presentation is impossible. Therefore, information about the economic system is presented in an aggregated manner.
Let's go to this description. We need to describe the economic system with a smaller number of so-called pure industries $n<m$. We introduce the denotation: $M=\{1,2, \ldots,m\}$ $N=\{1,2, \ldots,n\},$ and the mapping $f$ of the set $M$ onto $N.$ In the production process, an strictly positive price vector $p=\{p_i\}_{i=1}^m$ is formed between producers who have entered into contracts with each other for the supply of goods. It is formed under the influence of supply and demand for goods in the economic system and under the influence of the formed system of both direct and indirect taxes.
By vectors
$x=\{x_i\}_{i=1}^m$, $p=\{p_i\}_{i=1}^m,$ and the matrix $ A $ we introduce the following quantities:
\begin{eqnarray}\label{agr2}
 X_k= X_k(p,x)=\sum\limits_{\{l, f(l)=k\}}p_lx_l, \quad C_k=C_k(p,x)=\sum\limits_{\{l, f (l)=k\}}p_lc_l, \quad k=\overline{1,n}.
\end{eqnarray}
\begin{eqnarray}\label{agr3}
\bar a_{ki}=\bar a_{ki}(p,x)=\frac{\sum\limits_{\{s, f(s)=k\}}\sum\limits_{\{r, f (r)=i\}} p_s a_{sr}x_r}{ \sum\limits_{\{l, f(l)=i\}}p_lx_l}, \quad k,i=\overline{1,n} .
\end{eqnarray}
\begin{prope}\label{agreg1}
Let $\hat p=\{\hat p_i\}_{i=1}^n$ $\hat X=\{\hat X_i\}_{i=1}^n$ be some strictly positive vectors from  $ R_+^n,$ and $\ p=\{ p_i\}_{i=1}^m$ and $x=\{ x_i\}_{i=1}^m$ are strictly positive vectors of prices and gross outputs from the subset  $R_+^m$, respectively.
According to the aggregation mapping $f(i), i \in M,$
let us consider the vectors $\hat p(p)=\{\hat p_f(i) p_i\}_{i=1}^m,$ $\hat X(x)=\{\hat X_f(i) x_i \}_{i=1}^m$ from the subset $R_+^m.$ Then the  equalities
\begin{eqnarray}\label{agreg2}
\bar a_{ki}(\hat p(p),\hat X(x)) =\frac{\hat p_k\bar a_{ki}(p,x)}{\hat p_i},
\quad k,i=\overline{1,n},
 \end{eqnarray}
\begin{eqnarray}\label{agreg3}
X_k(\hat p(p), \hat X(x))=\hat p_k \hat X_k X_k(p,x), \quad k=\overline{1,n},
 \end{eqnarray}
 \begin{eqnarray}\label{agreg4}
C_k(\hat p(p), \hat X(x))=\hat p_k \hat X_k C_k(p,x), \quad k=\overline{1,n},
 \end{eqnarray}
 are valid.
\end{prope}
\begin{proof}
The proof of Proposition \ref{agreg1} is evident.
\end{proof}
Strictly positive vectors $\hat p=\{\hat p_i\}_{i=1}^n$ $\hat X=\{\hat X_i\}_{i=1}^n$ from the subset $R_+^ n$ is called the relative price vector and the relative gross output vector, respectively.
\begin{prope}\label{agre3}
If the matrix $A=||a_{ki}||_{k,i=1}^m$ is productive, then the aggregated matrix $\bar A=||\bar a_{ki}||_{k ,i=1}^n$ is also productive.
\end{prope}
\begin{proof} Due to productivity of  matrix
$A$ for any strictly positive vector $c=\{c_i\}_{i=1}^m$ there exists a strictly positive solution $x=\{x_i\}_{i=1}^m$ to the systems of equations (\ref{agr1}). Then the aggregated strictly positive vector $X=\{X_i\}_{i=1}^n$ is a solution of the system of equations (\ref{agr4}) with the strictly positive vector $C=\{C_i\}_{i =1}^n.$ The latter means the productivity of the matrix $\bar A=||\bar a_{ki}||_{k,i=1}^n.$
\end{proof}

For the price vector $p=\{p_i\}_{i=1}^m,$ that was formed under the action of competitive forces that are the result of the taxation system, production technologies, monetary policy of the central bank, we introduce into consideration  the vector $ \delta=\{\delta_i\}_{i=1}^n,$
where
\begin{eqnarray}\label{agr5}
\delta_i=p_i-\sum\limits_{s=1}^ma_{si}p_s
\quad i=\overline{1,m}.
\end{eqnarray}
The vector $\delta=\{\delta_i\}_{i=1}^n$ is called the vector of created added values. Under the conditions described above, this vector may not be strictly positive. This means that not all production technologies are capable of creating positive added value under competitive conditions. But as a result of aggregation, it may turn out that the aggregated value added $\Delta_j=\sum\limits_{\{l, f(l)=j\}}\delta_l p_l, \ j=\overline{1,n},$ can  be both strictly positive and not  be strictly positive in all industries.

Since we consider that the vector $x=\{x_i\}_{i=1}^m$ is a solution of the system of equations (\ref{agr1}), then the vector $X=\{X_i\}_{i= 1}^n$ satisfies the system of equations
\begin{eqnarray}\label{agr4}
X_k-\sum\limits_{i=1}^n \bar a_{ki}X_i=C_k, \quad k=\overline{1,n}.
\end{eqnarray}
It is obvious that the following equalities 
\begin{eqnarray}\label{agr6}
X_j-\sum\limits_{s=1}^n \bar a_{sj}X_j=\Delta_j, \quad j=\overline{1,n},
\end{eqnarray}
is true, where $\Delta_j=\sum\limits_{\{l, f(l)=j\}}\delta_l p_l, \ j=\overline{1,n}.$
\begin{de}\label{nok1}
The description of  economy in value indicators using the gross output vector $X=\{X_i\}_{i=1}^n,$
 vector of final consumption
  $C=\{C_i\}_{i=1}^n,$ vector of added values 
 $\Delta=\{\Delta_i\}_{i=1}^n$ and matrix of direct costs $\bar A$=$||a_{ki}||_{k,i=1}^n,$ which satisfy the relations (\ref{agr4}) and (\ref{agr6}) is called aggregated one.
A consequence of the systems of equations (\ref{agr4}) and (\ref{agr6}) is the equality
\begin{eqnarray}\label{nok2}
\sum\limits_{k=1}^n C_k=\sum\limits_{s=1}^n
\Delta_s,
\end{eqnarray}
which means that the sum of the values of final product created in all industries  is equal to the sum of added values created in all industries.
\end{de}

\begin{prope}\label{agre4}
If $\Delta_j>0, \ j=\overline{1,n},$ then the inequalities $\sum\limits_{s=1}^n \bar a_{sj}<1,\ j=\overline{1 ,n},$
are true.
\end{prope}
\begin{proof}
The proof follows immediately from the system of equalities (\ref{agr6}).
\end{proof}
Consider the system of equations in relation to the relative price vector
$\hat p=\{ \hat p_i\}_{i=1}^n,$ which is the solution of the system of equations
\begin{eqnarray}\label{agr7}
\hat p_i= \sum\limits_{s=1}^n\hat p_s \bar a_{si} +\hat \delta_i, \quad i=\overline{1,n}
\end{eqnarray}
where $\hat \delta_i>0, \ i=\overline{1,n}.$
Due to the productivity of the matrix
$ \bar A $ there is a strictly positive solution of the system of equations (\ref{agr7}).
\begin{prope}\label{agrega1}
Let $p=\{p_i\}_{i=1}^m$ be a strictly positive price vector from the subset $R_+^m,$ and $\hat p=\{\hat p_i\}_{i=1}^ n$ is a strictly positive relative price vector from  $R_+^n,$ that satisfies the system of equations (\ref{agr7}).  Under the condition of aggregation $f(i), i \in M,$ let
\begin{eqnarray}\label{agrega2}
\delta_i (\hat p)=\hat p_{f(i)} p_i - \sum\limits_{s=1}^m \hat p_{f(s)} p_s a_{si}, \quad i=\overline{1,m}
\end{eqnarray}
be created added values for the price vector
$\hat p(p).$
Then the following equalities 
\begin{eqnarray}\label{agrega3}
\sum\limits_{\{i, f(i)=k\}}\delta_i (\hat p)x_i=X_k \hat \delta_k, \quad k=\overline{1,n},
\end{eqnarray}
are hold, where $ x=\{x_i\}_{i=1}^m$ is a strictly positive gross output vector, and $X_k=\sum\limits_{\{i, f(i)=k\}}x_i p_i.$
\end{prope}

\begin{proof}
It is obvious that
$$ \sum\limits_{\{i, f(i)=k\}}\delta_i (\hat p)x_i=
\sum\limits_{\{i, f(i)=k\}}p_i \hat p_{f(i)}x_i -\sum\limits_{\{i, f(i)=k\}}\sum \limits_{s=1}^m \hat p_{f(s)} p_s a_{si}x_i=$$
$$\hat p_k X_k -
\sum\limits_{\{i, f(i)=k\}}\sum\limits_{l=1}^n \sum\limits_{\{s, f(s)=l\}} \hat p_ {l} p_s a_{si}x_i=$$
$$\hat p_k X_k-
\sum\limits_{l=1}^n \sum\limits_{\{i, f(i)=k\}} \sum\limits_{\{s, f(s)=l\}} \hat p_ {l} p_s a_{si}x_i=
$$
$$\hat p_k X_k-
\sum\limits_{l=1}^n \frac{\hat p_{l} \bar a_{lk}}{\hat p_k} \hat p_k X_k =\hat p_k X_k(1-\sum\limits_{l =1}^n \frac{\hat p_{l} \bar a_{lk}}{\hat p_k})=X_k \hat \delta_k, \quad k=\overline{1,n}.
$$
\end{proof}
\begin{ce}\label{agrega4}
Let the aggregated added values  $\Delta_j>0, \  \ j=\overline{1,n}. $  If $\hat \delta_j=\frac{\Delta_j}{X_j}, \j=\overline{1,n},$
then
$$ \sum\limits_{\{i, f(i)=j\}}\delta_i (\hat p)x_i=\Delta_j, \quad j=\overline{1,n}.$$
In this case, the relative equilibrium price vector $\hat p=\{\hat p_i\}_{i=1}^n
$ is such that $\hat p_i=1, i=\overline{1,n}.$
\end{ce}
\begin{proof}
Consider the system of equations in relation to the relative price vector
$\hat p=\{ \hat p_i\}_{i=1}^n,$ which is the solution of the system of equations
\begin{eqnarray}\label{atina1}
\hat p_i= \sum\limits_{s=1}^n\hat p_s \bar a_{si} +\frac{\Delta_i}{X_i}, \quad i=\overline{1,n}.
\end{eqnarray}
Due to the fact that equalities  (\ref{agr6}) hold,
the system of equations (\ref{atina1}) can be written in the form
\begin{eqnarray}\label{atina2}
\hat p_i= \sum\limits_{s=1}^n\hat p_s \bar a_{si} +1-\sum\limits_{s=1}^n \bar a_{si}, \quad
 i=\overline{1,n}.
\end{eqnarray}
Due to the productivity of the matrix
$ \bar A $ strictly positive solution of the system of equations (\ref{atina2}) exists and is unique, and therefore it is equal to the unit vector.
\end{proof}
\begin{note}\label{agrega5}
If all created value added $\delta_j>0, \ i=\overline{1,m},$ for the price vector $p=\{p_i\}_{i=1}^m,$ that has formed in the economic system, then $\Delta_j>0, \j=\overline{1,n}. $
\end{note}
\begin{note}\label{agrega6}
It follows from Consequence \ref{agrega4} and Remark \ref{agrega5} that the quantities $\hat \delta_j, j=\overline{1,n},$ appearing in the system of equations (\ref{agr7}) should have structure $\hat \delta_j=\frac{\Delta_j}{X_j}+\hat \delta_j^1 , j=\overline{1,n},$ where $\hat \delta_j^1 \geq 0, j=\overline{1,n}.$ Therefore, the relative price vector should be considered as a corrector of prices that can be deformed by monopoly or oligopoly. The Antimonopoly Committee should collect fines from monopolistic entities for subsidizing unprofitable industries.
\end{note}
\begin{te}\label{agre5} Let the strictly positive relative price vector $\hat p=\{\hat p_i\}_{i=1}^n$ satisfy the system of equations
\begin{eqnarray}\label{axy1}
\hat p_i- \sum\limits_{s=1}^n\hat p_s \bar a_{si} =\hat \delta_i, \quad i=\overline{1,n},
\end{eqnarray}
where, $ \hat \delta_i>0, \ i=\overline{1,n}.$
Then there exists a taxation system $\pi=\{\pi_k\}_{k=1}^n,$ for which the relative price vector $\hat p$ is an equilibrium price vector, that is, such that the equalities
\begin{eqnarray}\label{agr8}
\sum\limits_{i=1}^n\bar a_{ki}\frac{(1-\pi_i)\hat p_i X_i}{\sum\limits_{s=1}^n\hat p_s \bar a_{ si}}=(1-\pi_k) X_k, \quad k=\overline{1,n},
\end{eqnarray}
are true.
\end{te}
\begin{proof}  
We introduce the following denotations $\hat X_i=X_i \hat p_i,$ $\hat a_{ki}=\frac{\hat p_k \bar a_{ki}}{\hat p_i}.$
In these denotations, the system of equations (\ref{agr8}) is written in the form
\begin{eqnarray}\label{ax1}
\sum\limits_{i=1}^n \hat a_{ki} \frac{(1-\pi_i)\hat X_i }{\sum\limits_{s=1}^n\hat a_{si}}= (1-\pi_k)\hat X_k, \quad k=\overline{1,n}.
\end{eqnarray}
Let
$$ V_i=\frac{\hat X_i(1-\pi_i)}{\sum\limits_{s=1}^n\hat a_{si}}, \quad i=\overline{1,n}.$$
Then the vector $V=\{V_i\}_{i=1}^n$ satisfies the system of equations
\begin{eqnarray}\label{ax2}
 \sum\limits_{i=1}^n \hat a_{ki}V_i=\sum\limits_{s=1}^n\hat a_{sk}V_k, \quad k=\overline{1,n}.
\end{eqnarray}
On the basis of Lemma \ref{pupvittin7}, there exists a strictly positive solution of the system of equations (\ref{ax2}), which is determined with accuracy up to a constant. Let us denote $V_0=\{V_i^0\}_{i=1}^n$
the solution of this system of equations, the sum of the components of which is equal to one. Then
$V=\{V_i\}_{i=1}^n=c_0V_0=\{c_0 V_i^0\}_{i=1}^n.$
From here
$$1-\pi_i=\frac{c_0V_i^0\sum\limits_{s=1}^n\hat a_{sk}}{\hat X_i} =c_0\frac{V_i^0}{\hat X_i}( 1-\frac{\hat \Delta_i}{\hat X_i}),\quad i=\overline{1,n},$$
where it is taken into account that the relation takes place
$$ \hat X_i=\sum\limits_{s=1}^n\hat a_{si} \hat X_i +\hat \Delta_i, \quad i=\overline{1,n}.$$

If we take into account the relation
$$ \hat \Delta_i=X_i\hat \delta_i \, \quad i=\overline{1,n}, $$
on the basis of Proposal \ref{agrega1}, then
we  get

$$\pi_i=1-c_0\frac{V_i^0}{\hat X_i}(1-\frac{\hat \Delta_i}{\hat X_i})=$$
\begin{eqnarray}\label{ax3}
1-c_0\frac{V_i^0}{\hat X_i}(1-\frac{X_i\hat \delta_i}{\hat X_i})=1-c_0\frac{V_i^0}{\hat X_i}( 1-\frac{\hat \delta_i}{\hat p_i}), \quad i=\overline{1,n}.
\end{eqnarray}
The constant $c_0>0$ can be chosen such that the inequalities are fulfilled
$$ 1-c_0\frac{V_i^0}{\hat X_i}(1-\frac{\hat \Delta_i}{\hat X_i})>0, \quad i=\overline{1,n}. $$
The theorem \ref{agre5} is proved.
\end{proof}

\begin{ce}\label{tinaKagy1}
The best taxation system $\pi=\{\pi_i\}_{i=1}^n$ under the condition that the final product will be created in the economic system is such that satisfies the equality
\begin{eqnarray}\label{tinaKagy2}
\frac{1-\pi_i}{\frac{V_i^0}{\hat X_i}(1-\frac{\hat \delta_i}{\hat p_i})}=\frac{1}{\max\limits_ {1\leq i\leq n}\frac{V_i^0}{\hat X_i}(1-\frac{\hat \delta_i}{\hat p_i})}, \quad i=\overline{1,n }.
\end{eqnarray}
\end{ce}
\begin{proof}
We choose the constant $c_0$ in the formula (\ref{ax3}) so that the value of the values 
$\pi_i, \ i=\overline{1,n}, $  is the smallest. For this, we should put
\begin{eqnarray}\label{tinaKagy3}
c_0 =\frac{1}{\max\limits_{1\leq i\leq n}\frac{V_i^0}{\hat X_i}(1-\frac{\hat \delta_i}{\hat p_i}) }.
\end{eqnarray}
Then equalities (\ref{tinaKagy2}) will hold.
\end{proof}

In the following Theorem \ref{mainagy1} we find out the conditions for the taxation system $\pi=\{\pi_i\}_{i=1}^n,$ that has developed in the economic system, under which the final product will be created such that the economic system would be able to function in
 the mode of sustainable development under the equilibrium price vector.
 
 \begin{te}\label{mainagy1}
 Let the matrix $ \bar A=||\bar a_{ki}||_{k,i=1}^n$ be indecomposable and productive. Then for the relative price vector $\hat p=\{\hat p_i\}_{i=1}^n,$ which is a solution of the system of equations (\ref{axy1}), there is always a solution of the system of equations ( \ref{agr8}) with respect to the vector $X=\{X_i\}_{i=1}^n,$ for which the final product will be created, that is, which will satisfy the system of equations (\ref{agr4})
with the strictly positive right-hand side of this system of equations, provided that the taxation system
$\pi=\{ \pi_i\}_{i=1}^n,$ satisfies the conditions
\begin{eqnarray}\label{mainagy2}
 0 < \pi_i\leq 1- b(1-\frac{\hat\delta_i}{\hat p_i}), \quad i=\overline{1,n},
\end{eqnarray}
for some $b,$ that satisfies the inequalities
$$ \max\limits_{1 \leq i \leq n}(1- \pi_i)<b < \frac{1}{ \max\limits_{1 \leq i \leq n}(1-\frac{ \hat \delta_i}{\hat p_i})}. $$
 \end{te}
 \begin{proof}
 Let's put $\hat X_k=X_k\hat p_k, $ $\hat A=|\hat a_{ki}|_{k, i=1}^n, \hat a_{ki}=\frac {\hat p_k \bar a_{ki}}{\hat p_i},$
 In these denotations, the system of equations (\ref{agr8}) can be written in the form
\begin{eqnarray}\label{mainagy3}
\sum\limits_{i=1}^n \hat a_{ki} \frac{(1-\pi_i)\hat X_i }{\sum\limits_{s=1}^n\hat a_{si}}= (1-\pi_k)\hat X_k, \quad k=\overline{1,n},
\end{eqnarray}
and the system of equations (\ref{agr6}) is in the form
\begin{eqnarray}\label{mainagy7}
\hat X_i-\sum\limits_{s=1}^n \hat a_{si} \hat X_i=\hat \Delta_i, \quad i=\overline{1,n}.
\end{eqnarray}
By introducing the denotation
\begin{eqnarray}\label{mainagy4}
\frac{(1-\pi_i)\hat X_i }{\sum\limits_{s=1}^n\hat a_{si}}=\hat X_i^0, \quad i=\overline{1,n} ,
\end{eqnarray}
 the system of equations (\ref{mainagy3}) is rewritten in the form
 \begin{eqnarray}\label{mainagy5}
\sum\limits_{i=1}^n \hat a_{ki} \hat X_i^0 = \sum\limits_{s=1}^n\hat a_{sk} \hat X_k^0, \quad k=\overline {1,n}.
\end{eqnarray}
On the basis of Lemma \ref{pupvittin7}, there is a strictly positive solution of this system of equations, which is determined with accuracy up to a constant.
Let $\hat X_0=\{\hat X_i^0\}_{i=1}^n$ be a strictly positive solution of the system of equations (\ref{mainagy5}), the sum of whose components is equal to unity. Then any other is equal
$c_0\hat X_0=\{c_0\hat X_i^0\}_{i=1}^n$, where $c_0>0.$
For the components of the vector $\hat X=\{\hat X_i\}_{i=1}^n$
we get the formulas
 \begin{eqnarray}\label{main6}
 \hat X_i =c_0 \frac{\sum\limits_{s=1}^n\hat a_{si}}{1-\pi_i} \hat X_i^0=c_0\frac{(1-\frac{\hat \Delta_i}{\hat X_i})}{1-\pi_i}\hat X_i^0=c_0\frac{(1-\frac{\hat \delta_i}{\hat p_i})}{1-\pi_i} \hat X_i^0, \quad i=\overline{1,n}.
\end{eqnarray}
To prove that the obtained solution is a vector of gross output for some vector of final consumption, it is necessary to establish that the inequalities
\begin{eqnarray}\label{mainagy7}
\hat X_k-\sum\limits_{i=1}^n \hat a_{ki} \hat X_i>0, \quad k=\overline{1,n},
\end{eqnarray}
are  valid.
Due to the fact that the inequalities (\ref{mainagy2}) are valid, the  inequalities
\begin{eqnarray}\label{mainagy8}
1 \leq \frac{1-
\pi_i}{b(1-\frac{\hat \Delta_i}{\hat X_i})}=\frac{1-
\pi_i}{b\sum\limits_{s=1}^n \hat a_{si} }
\end{eqnarray}
are true, therefore
$$\hat X_k-\sum\limits_{i=1}^n \hat a_{ki} \hat X_i \geq \hat X_k-\frac{1}{b}\sum\limits_{i=1}^n \hat a_{ki}\frac{(1- 
\pi_i) \hat X_i}{\sum\limits_{s=1}^n \hat a_{si} }=$$
\begin{eqnarray}\label{mainagy9}
(1-\frac{1-\pi_k}{b} )\hat X_k>0, \quad k=\overline{1,n}.
\end{eqnarray}
So,
\begin{eqnarray}\label{mainagy10}
\hat X_k-\sum\limits_{i=1}^n \hat a_{ki} \hat X_i=\hat Y_k, \quad k=\overline{1,n},
\end{eqnarray}
where the vector $\hat Y=\{\hat Y_k\}_{k=1}^n$ has strictly positive components, and $\hat Y_k=\hat p_k Y_k,\ k=\overline{1,n}. $

From here we immediately get that the vector $ X=\{ X_i\}_{i=1}^n$ is the gross output vector for the tax system under consideration
\begin{eqnarray}\label{mainagy11}
 X_k-\sum\limits_{i=1}^n \bar a_{ki} X_i= Y_k, \quad k=\overline{1,n}.
\end{eqnarray}
 The theorem \ref{mainagy1} is proved.
\end{proof}
Below we reformulate Theorem \ref{agre5} and Corollary \ref{tinaKagy1} from it for the case when $\Delta_i >0, \i=\overline{1,n}$ and $\hat \delta_i=\frac{\Delta_i }{X_i}\ i=\overline{1,n}.$

\begin{te}\label{atina3}
Let $\Delta_i >0, \ i=\overline{1,n},$ and $\hat \delta_i=\frac{\Delta_i}{X_i},\ i=\overline{1,n}.$
Then there is a taxation system $\pi=\{\pi_k\}_{k=1}^n,$ under which the vector of prices $ p=\{p_i\}_{i=1}^m,$ which developed in the economic system, is an equilibrium price vector, i.e. such that equalities 
\begin{eqnarray}\label{atina4}
\sum\limits_{i=1}^n\bar a_{ki}\frac{(1-\pi_i) X_i}{\sum\limits_{s=1}^n \bar a_{si}}=(1 -\pi_k) X_k, \quad k=\overline{1,n},
\end{eqnarray}
are true.
\end{te}
The consequence of Theorem \ref{atina3} is
\begin{ce}\label{atina5}
Under the conditions of Theorem \ref{atina3}, the best taxation system $\pi=\{\pi_i\}_{i=1}^n$ is given by the formulas
\begin{eqnarray}\label{atina6}
\frac{1-\pi_i}{\frac{V_i^0}{ X_i}(1-\frac{\Delta_i}{X_i})}=\frac{1}{\max\limits_{1\leq i\leq n}\frac{V_i^0}{X_i}(1-\frac{ \Delta_i}{X_i})}, \quad i=\overline{1,n}.
\end{eqnarray}
where the vector $V_0=\{V_i^0\}_{i=1}^n$ is the solution of the system of equations
\begin{eqnarray}\label{atina7}
 \sum\limits_{i=1}^n \bar a_{ki}V_i^0=\sum\limits_{s=1}^n\bar a_{sk}V_k^0, \quad k=\overline{1,n},
\end{eqnarray}
and is such that the sum of its components is equal to one.
\end{ce}
\begin{te}\label{atina8}
 Let the matrix $ \bar A=||\bar a_{ki}||_{k,i=1}^n$ be indecomposable and productive,
 and let $\Delta_i >0, \ i=\overline{1,n},$ and $\hat \delta_i=\frac{\Delta_i}{X_i},\ i=\overline{1,n}.$
 There is always a solution to the system of equations
  \begin{eqnarray}\label{atina9}
\sum\limits_{i=1}^n\bar a_{ki}\frac{(1-\pi_i) X_i}{\sum\limits_{s=1}^n \bar a_{si}}=(1 -\pi_k) X_k, \quad k=\overline{1,n}.
\end{eqnarray}
with respect to the vector $X=\{X_i\}_{i=1}^n,$ for which the final product will be created, that is, which will satisfy the system of equations (\ref{agr4})
with the strictly positive right-hand side of this system of equations, provided that the taxation system
$\pi=\{ \pi_i\}_{i=1}^n,$ satisfies the conditions
\begin{eqnarray}\label{atina10}
 0 < \pi_i\leq 1- b(1-\frac{\Delta_i}{X_i}), \quad i=\overline{1,n},
\end{eqnarray}
for some $b,$ that satisfies the inequalities
$$ \max\limits_{1 \leq i \leq n}(1- \pi_i)<b < \frac{1}{ \max\limits_{1 \leq i \leq n}(1-\frac{  \Delta_i}{X_i})}. $$
 \end{te}
\section{Conclusions.}
In the work, a method of researching real economic systems for the possibility of their functioning in the mode of sustainable development is built. Section 2 develops algorithms for constructing equilibrium states for the case of partial market clearing. For this purpose, the vector of real consumption, which satisfies the system of linear equations and inequalities, is considered.
A complete description of the solutions of such a system of linear inequalities and equations is given. On this basis, a complete description of the equilibrium states under which partial clearing of the markets takes place in the production model "input - output" is given. These states are characterized by the level of excess supply in a certain period of the economy function.
In order to find equilibrium states with the lowest level of excess supply, it is proved that the vector of real consumption is the solution of a certain problem of quadratic programming. Section 3 establishes the necessary and sufficient conditions for the operation of the economic system in the mode of sustainable development.
It is shown that there is a family of taxation vectors under which the economic system described by the production model "input - output" functions in the mode of sustainable development. A restriction was found for taxation systems in real economic systems, under which the final product is created in the economic system, and the economic system functions in the mode of sustainable development.
Section 5 contains the axiomatic of the aggregated description of the economy. On this basis, an aggregated matrix of direct costs, an aggregated vector of gross output, an aggregated gross added value, and an aggregated vector of final consumption were introduced. The main Theorems about the existence of the taxation system and limitations for the taxation system are also formulated for the aggregated description of the economy.



\begin{thebibliography}{99}

\bibitem{11Gonchar11} 
Gonchar, N.S. (2024) Economy Function in the Mode of Sustainable Development. Advances in Pure Mathematics, 14, 242-282.
https://doi.org/10.4236/apm.2024.144015

\bibitem{10Gonchar2}   
Gonchar, N.S. (2023) Economy Equilibrium and Sustainable Development. Advances
in Pure Mathematics, 13, 316-346. https://doi.org/10.4236/apm.2023.136022

\bibitem{11Gonchar2}  
Gonchar, N.S. (2023) Mathematical Foundations of Sustainable Economy Development.
Advances in Pure Mathematics , 13, 369-401.
https://doi.org/10.4236/apm.2023.136024

\bibitem{Gonchar2}   Gonchar, N. S. (2008) Mathematical foundations of information economics.  Bogolyubov Institute for Theoretical  Physics, Kiev, 468p.

\bibitem{4Gonchar}  Gonchar, N.S. and Zhokhin, A.S. (2013) Critical States in Dynamical Exchange Model
and Recession Phenomenon. Journal of Automation and Information Science , 45,
50-58. https://doi.org/10.1615/JAutomatInfScien.v45.i1.40

\bibitem{3Gonchar} N.S. Gonchar, A.S. Zhokhin, W.H. Kozyrski,
(2015). General Equilibrium and Recession Phenomenon, 
American Journal of Economics, Finance and 
Management, 1: 559-573. 

\bibitem{5Gonchar}  Gonchar, N.S., Zhokhin, A.S. and Kozyrski, W.H. (2015) On Mechanism of Recession
Phenomenon. Journal of Automation and Information Sciences , 47, 1-17.
https://doi.org/10.1615/JAutomatInfScien.v47.i4.10




\bibitem{7Gonchar} Gonchar, N.S., Dovzhyk, O.P., Zhokhin, A.S., Kozyrski, W.H. and Makhort, A.P.
(2022) International Trade and Global Economy. Modern Economy , 13, 901-943.
https://www.scirp.org/journal/me
https://doi.org/10.4236/me.2022.136049

\bibitem{Nirenberg} Nirenberg L. (1974) Topics in nonlinear functional analysis. Courant Institute of Mathematical Sciences, New York University.

\bibitem{6Gonchar}  N. S. Gonchar, W.H. Kozyrski, A.S. Zhokhin, O.P. Dovzhyk, (2018).  Kalman Filter in the Problem of the Exchange and the Inflation Rates Adequacy to Determining Factors. Noble International Journal of Economics and Financial Research.  vol. 3 (3), pp.31-
39.

\bibitem{2Gonchar} Gonchar, N.S., Zhokhin, A.S. and Kozyrski, W.H. (2020) On Peculiarities of Ukrainian
Economy Development. Cybernetics and Systems Analysis , 56,  No 3,  439-448.
https://doi.org/10.1007/s10559-020-00259-0



\bibitem{1Gonchar} N.S. Gonchar, O.P. Dovzhyk. (2019)  On one criterion for the permanent economy development, Journal of Modern Economy, 2:9, pp. 1-16. . https://doi.org/10.28933/jme-2019-09-2205

\bibitem{8Gonchar} Gonchar, N.S. and Dovzhyk, O.P. (2022) On the Sustainable Economy Development
of Some European Countries. Journal of Modern Economy , 5, 1-14.
https://doi.org/10.28933/jme-2021-12-0505


\end{thebibliography}
\end{document}